\documentclass[twocolumn,prb,aps,floatfix,amsmath,amssymb,superscriptaddress]{revtex4}
\pdfoutput=1

\usepackage{mathtools}
\usepackage{amsfonts}
\usepackage{amssymb}

\usepackage{graphicx}

\usepackage{amsthm}

\usepackage{euscript}
\usepackage{cprotect}
\usepackage[usenames,dvipsnames]{xcolor}
\usepackage[arc,color,frame,knot,pdf,poly]{xy}

\usepackage[strict]{changepage}

\newtheorem{theorem}{Theorem}[section]
\newtheorem{lemma}[theorem]{Lemma}

\newtheorem{proposition}[theorem]{Proposition}

\numberwithin{equation}{section}
\numberwithin{figure}{section}
\DeclareMathOperator{\Free}{Free}

\newcommand{\be}{\begin{equation}}
\newcommand{\ee}{\end{equation}}
\newcommand{\poly}{{\rm poly}}

\begin{document}
\title{Obstructions To Classically Simulating The Quantum Adiabatic Algorithm}

\author{M. B. Hastings$^1$,\\ with Appendix by M. H. Freedman}
\affiliation{Microsoft Research, Station Q, CNSI Building, University of California, Santa Barbara, CA, 93106}

\begin{abstract}
We consider the adiabatic quantum algorithm for systems with ``no sign problem", such as the transverse field Ising mode, and analyze the equilibration time for quantum Monte Carlo (QMC) on these systems. We ask: if the spectral gap is only inverse polynomially small, will equilibration methods based on slowly changing the Hamiltonian parameters in the QMC simulation succeed in a polynomial time?
We show that this is {\it not} true, by constructing  counter-examples. In some examples, the space of configurations
where the wavefunction has non-negligible amplitude has a nontrivial fundamental group, causing
the space of trajectories in imaginary time to break into disconnected components with only negligible probability outside these components.   
 For the simplest example we give with an abelian fundamental group, QMC does not equilibrate but still solves the optimization problem.  More
severe effects leading to failure to solve the optimization can occur when the fundamental group is a free group on two generators.  
Other examples where QMC fails have a {\it trivial} fundamental group, but still use ideas from topology relating group presentations to simplicial complexes.  We define gadgets to realize these Hamiltonians as the effective low-energy dynamics of a transverse field Ising model.
We present some analytic results on equilibration times which may be of some independent interest in the theory of equilibration of Markov chains.  Conversely, we show that a small spectral gap implies slow equilibration at low temperature for some initial conditions and for a natural choice of local QMC updates.
\end{abstract}
\maketitle

The quantum adiabatic algorithm\cite{adqopt} uses a parameter-dependent Hamiltonian $H_s$ to find the ground state of a classical optimization algorithm.
The Hamiltonian $H_0$ is chosen to have a simple ground state, typically of product form, that can be easily prepared.  After initializing the system in this state, the system is evolved under a time-dependent Hamiltonian, $H_{s(t)}$, with $s(t)$ slowly changing as a function of time $t$, until at some final time $t_f$ with $s(t_f)=1$, the Hamiltonian describes some classical optimization problem. 
If the rate of change is sufficiently slow, then the system stays close to the ground state throughout this procedure, thus solving the optimization problem.
The original adiabatic quantum algorithm relied on coherent evolution with a sufficiently large gap $\Delta$ between the ground and first excited state along the path of $H_s$.  The time required for the algorithm scales $1/\Delta^2$, and depends also upon the norm of terms in the Hamiltonian.  Later work considered effects of incoherent evolution\cite{childs}.
One potential advantage of the adiabatic algorithm is that it is a general purpose tool; similarly to the classical simulated annealing algorithm, while the quantum adiabatic algorithm may not be the best tool for any given problem, it can be applied to a wide range of problems.

Unfortunately, little is known definitely about the performance of the adiabatic quantum algorithm for practical problems, despite much work on studying small systems using exact diagonalization\cite{ed} and larger systems using numerical quantum Monte Carlo (QMC) studies\cite{apy}.  One remarkable recent study\cite{troyer} is based directly on studying the Hamiltonian of the D-Wave device\cite{dwave}, showing evidence for nontrivial collective quantum effects.
The QMC studies, and this device, all involve Hamiltonians with ``no sign problem", as explained below; in particular, they are transverse field Ising models.

While the quantum adiabatic algorithm with arbitrary Hamiltonians is known to be as powerful as the circuit model for quantum computation\cite{equivalence}, it is unclear what advantage such Hamiltonians with no sign problem can have over classical computation.  Perhaps one can always simulate such systems using QMC?  In particular, since the main theoretical problem with QMC is understanding the equilibration time, perhaps a bound on the spectral gap $\Delta$ implies some bound on the equilibration time, if we follow certain protocols for equilibrating the QMC simulation?  In this paper, we show that for the most natural equilibration protocol, where one equilibrates the QMC at $s=0$ and then slowly changes $s$ and tries to equilibrate the QMC $s$ changes, this is {\it not} true in general.  We give several counter-examples (we modify the algorithm and protocol to account for some of the counter-examples, but then find other counter-examples to the modified algorithm).

One reason that we might expect this annealing protocol to work is that a similar approach {\it does} work when using a matrix product state algorithm (instead of QMC) to study a one-dimensional quantum system with a spectral gap, in that we exploit the idea of following the matrix product state along the path to avoid problems with getting trapped in a local minimum\cite{1dadiabatic}.  That result exploits the area law\cite{area1,area2}.
Another result along these lines is that if we restrict to
{\it frustration-free} Hamiltonians without a sign problem, then the problem
of simulating adiabatic evolution is in the complexity class BPP\cite{t1}.  Note
also that if we leave the context of adiabatic computing, the problem
of approximating the ground state energy of sign-problem free
Hamiltonians is in\cite{t2} complexity class
AM, while the analogous problem for arbitrary local Hamiltonans
is QMA-complete\cite{Kitaev}.

For a further illustration of why we might believe that the spectral gap $\Delta$ is related to the equilibration time of quantum Monte Carlo, consider
a single particle Hamiltonian in a tilted double-well potential:
\be
\label{doublewell}
H=-\frac{1}{2m}\partial^2+\mu x^2+x^4+hx,
\ee
with $\mu$ negative.  Suppose we start at positive $h$, where the ground state has most of its amplitude in the left well, and then change the sign of $h$,
trying to equilibrate the QMC as $h$ is changed.  Problems with equilibration may occur for large $|\mu|$, where the barrier between the wells is large, because the trajectory of the particle in imaginary time can get stuck in the left well, unable to tunnel through the barrier.  However, if the barrier is indeed high, then at $h=0$, the spectral gap becomes small.  So, based on this simple example one might expect a connection between spectral gap and equilibration of QMC.  We emphasize that the connection between spectral gap and equilibration only appears here when we follow this particular annealing protocol; if we instead start at  non-negligible $h$ but with the initial condition of the particle in the wrong well, then the QMC might be slow to equilibrate even though there is no spectral gap.

However, this example is really based on having the space of likely positions of the particle split into two disconnected sets, the left and right wells, with the coordinate $x=0$ being unlikely to occur.  That is, one may say this is an example of a nontrivial $\pi_0$, the zeroth homotopy group, of the space of likely positions.  QMC, however, considers a {\it trajectory} in imaginary time, and so one might expect obstructions based on a nontrivial $\pi_1$, the fundamental group.  This is, in fact, the basis for some of our examples below.

In a sense, this kind of obstruction is well-known.  For example, problems with equilibrating different winding number sectors when simulating particles on a torus are well-studied and various nonlocal update rules have been introduced to try to alleviate this problem\cite{winding}.
However, these nonlocal updates are often introduced in a way that is quite specific to the particular Hamiltonian considered and we do not know a general way to implement them that would deal with the cases considered later, so we do not consider nonlocal updates further.

Further, while problems with equilibrating different topological sectors have been considered before, typically this has been studied for abelian fundamental groups such as the fundamental group of the circle or torus while we consider cases where the fundamental group is a free group on two generators.  The non-abelian nature of this group leads to significantly worse effects from the different winding number sectors, as discussed below. 

After giving these counter-examples with a nontrivial fundamental group, we show how to realize them using gadgets.  Then, we give counter-examples with a {\it trivial} fundamental group.  Topology still plays a role in these examples, as we relate a certain presentation of the trivial group to a simplicial complex.  These examples build on the examples with nontrivial fundamental group, so that section should be read first.

Finally, we present more analytic results.  We show that the converse of the conjecture is true: a small spectral gap implies slow equilibration for certain initial conditions.  We also present some upper bounds on the equilibration time.

An appendix at the end of the paper, due to M. H. Freedman, provides some additional geometric and topological context to phenomena discussed in the main text.

\section{Review of Adiabatic Algorithm, Sign Problem, and the Annealing Protocol}
The adiabatic algorithm considers a parameter dependent Hamiltonian $H_s$.
A typical application of this algorithm would be a system of $N$ spin-$1/2$ spins with a parameter dependent Hamiltonian $H_s$ of the form
\be
\label{exQAD}
H_s=
-(1-s) \sum_i S^x_i +  s V,
\ee
where $i=1,...,N$ labels the different sites, $S^x_i$ is the $x$-component of the spin on site $i$ and $V$ is some operator which is diagonal in the $z$-basis.  The operator $V$ will typically be a sum of many terms, each depending upon a small number of spins.  For example, one could have $V=\sum_{i,j} J_{ij} S^z_i S^z_j$ for some matrix $J_{ij}$.  We will assume some polynomial bound on the norm of the terms in the Hamiltonian.
The goal of the algorithm is to find a configuration of the spins in the $z$ basis which minimizes $V$.  Such a configuration could be the 
solution of some classical optimization problem, with the particular problem being encoded in the matrix $J$.
Note that the ground state of $H_0$ is a product state with all spins polarized in the $x$-direction and so can be prepared easily.

The performance of the algorithm depends upon the minimum spectral gap $\Delta$ along the path.
As a minor technical point, for some optimization problems, we find that $V$ has a degenerate ground state.  In this case, the gap between the ground and first excited state goes to zero at the end of the path.  This does not pose a problem as any of the final states represents a solution to the optimization problem.  Alternately, suppose that for most of the path (namely, for the portion of the path with $B \geq 1/\poly(N)$), the inverse gap is at most polynomial.  Then, we run the adiabatic algorithm for the portion of the path with $B \geq 1/\poly(N)$), and then terminate.  Measuring the final state in the $z$-basis will give an outcome that is, with probability $1-1/\poly(N)$, a minimum of $V$.

\subsection{Sign Problem}
The Hamiltonian (\ref{exQAD}) is an example of a Hamiltonian with ``no sign problem", enabling the use of a path integral QMC algorithm explained below to study properties of the ground state.  We now give a fairly general explanation of the sign problem and of a simplified QMC algorithm.  Consider a Hamiltonian $H$ and a basis of state $\psi(c)$ where $c$ is some discrete index.  We refer to $c$ as a ``configuration".  Then, we say that ``the Hamiltonian has no sign problem" if whenever $c \neq d$ we have
\be
\label{nosignproblem}
\langle \psi(c), H \psi(d) \rangle \leq 0
\ee
Of course, this definition is somewhat imprecise.  For {\it any} Hamiltonian, we can find a basis transformation to diagonalize the Hamiltonian in which case Eq.~(\ref{nosignproblem}) is satisfied for that basis.  
However, for the QMC to be efficient, we want to be able to efficiently calculate $\langle \psi(c), H \psi(d) \rangle$ for all $c,d$.  To do this,
we are often interested in the case that the basis $\psi(c)$ is a product basis; i.e., we consider a system of $N$ sites, as in Eq.~(\ref{exQAD}), with the Hilbert space of the whole system being the tensor product of the $N$ different Hilbert spaces and the basis $\psi(c)$ should be a product basis.
It should be noted  that while Eq.~(\ref{nosignproblem}) is sufficient not to have a sign problem, it is not necessary; for example by a basis change, various Heisenberg models which seems to have a sign problem in one basis can be shown not to have a sign problem in another\cite{otherQMC}.

The simplest version of the path integral Monte Carlo works as follows.  Consider the partition function
$Z={\rm Tr}(\exp(-\beta H))$ for some given $\beta$.  We divide the system into $K$ ``time slices" for some integer $K$, introducing one index $c_i$ per time slice and summing over all indices.  We then write:
\begin{eqnarray}
\label{statweight}
&&{\rm Tr}(\exp(-\beta H))\\ \nonumber
&=&\sum_{\{c_i\}} \prod_{i=1}^{K} \langle \psi(c_{i+1}) | \exp(-\beta H/K) | \psi(c_{i}) \rangle,
\end{eqnarray}
where we have $i=1,...,K$ and where the sum is over all possible values of $c_1,...,c_K$.  We fix $c_{K+1}=c_1$; that is, $c_i$ is periodic.
We refer to a sequence of $c_1,...,c_K$ as a trajectory, saving the term ``path" instead for a path in parameter space.
We can then approximate the terms in the product by something which we can efficiently calculate
\begin{eqnarray}
&&
\langle \psi(c_{i+1}) | \exp(-\beta H/K) | \psi(c_{i}) \rangle \\ \nonumber
& \approx & \delta_{c_{i+1},c_i} - \frac{\beta}{K} \langle \psi(c_{i+1})| H | \psi(c) \rangle.
\end{eqnarray}
For large enough $K$ (polynomially large in $\beta$ and in the norm of $H$), the approximation error becomes negligible.

At this point, we now have expressed the partition sum as a sum over positive quantities which can be statistically sampled using a Monte Carlo procedure.  This allows one to directly determine observables which are diagonal in the basis such as $Z^{-1} \langle \psi(c) | \exp(-\beta H) | \psi(c) \rangle$ by measuring the probability distribution of the indices $c_i$.  Using more sophisticated methods, it is possible to determine off-diagonal quantities; for example, see Ref.~\onlinecite{offdiagworm}.

The discretization described here is somewhat simplified.  In practice, more sophisticated methods are often used to deal with discretization in a more efficient way\cite{continuous}.

Part of the description of a quantum Monte Carlo Algorithm is not just the space of states (in this case, the sequence of $c_1,...,c_K$) and the probabilities (given above) but also the transition rule.  The simplest possible choice is a local update rule in which we randomly pick a given $c_i$ and then try randomly changing the value of that $c_i$.  For a system with an exponentially large number of possible choices of $c_i$ (as in (\ref{exQAD}), where each $c_i$ takes one of $2^N$ possible values), we often consider changing only the value of one or a small number of spins at a time.

In fact, such a choice of transition rule is essential to defining what we mean by $\pi_1$.
If we have a discrete set of states $\psi(c)$, we can define a graph, with vertices of the graph corresponding to possible values of $c$, and edges
between vertices $c,d$ if $\langle \psi(c), H \psi(d) \rangle \neq 0$.
To define the concept of $\pi_1$, we interpret the graph as a $1$-complex, and we add some $2$-cells to the $1$-complex corresponding to different local updates that the QMC can implement: for example, if we have $3$ different vertices, $c,d,e$, with edges connecting all three, and it is possible for a local update to change the sequence $c_1=c, c_2=d, c_3=e$ into the sequence $c_1=c, c_2=e, c_3=e$, then we attach a $2$-cell to those three $1$-cells.  If we can update a sequence $c_1=c, c_2=d, c_3=e$ into the sequence $c_1=c, c_2=d', c_3=e$ then we attach a $2$-cell to the four $1$-cells corresponding to following four edges of the graph: $(c,d),(d,e),(e,d'),(d',c)$.

However, we will not consider this kind of more formal definition of a complex any further.  The reason is, we would have to also in some way decide that configurations which appear with negligible amplitude should be removed from the graph when computing $\pi_0$ (as in Eq.~(\ref{doublewell}) or $\pi_1$.  Currently, we have not given a precise way of specifying how to remove these configurations.  Because of this imprecision, we will be content with heuristically arguing below that certain algorithms will be unable to equilibrate in certain cases.

\subsection{Annealing}
The most fundamental problem with QMC, as with any Monte Carlo algorithm, especially when applied to optimization, is the problem of equilibration.
At $s=0$, the system equilibrates rapidly,
but in general, we expect that equilibration for large enough $s$ may be exponentially slow for some choices of the Hamiltonian for at least {\it some} choices of the initial state.

However, there are many possible annealing protocols for the Monte Carlo dynamics.  In classical Monte Carlo simulations, for example, we can follow a simulated annealing procedure of slowly reducing the temperature.  Similarly, we can anneal the QMC dynamics using Hamiltonian $H=H_s$ by slowly increasing $\beta$ or by slowly changing $s$.
This last is the case that we analyze in this paper: fixing $\beta$ and slowly changing $s$.  The goal is to change $s$ sufficiently slowly that the procedure remains close to equilibrium throughout.  
A more sophisticated approach than slowly changing $s$ is to use a parallel tempering procedure (equilibrating at several different values of $s$ simultaneously, and allowing moves that swap trajectories between different values of $s$); we do not analyze this in this paper.
See Refs.~\onlinecite{apy,troyer,sandvikchanging} for practical implementations of some of these approaches.
Later we briefly consider the case in which $\beta$ is allowed to change during the annealing.

The conjecture, then, is that if we start at $s=0$ with the Monte Carlo procedure equilibrated and if the parameter $s$ changes only by a small amount $\epsilon$ from one step to the next, then if $\epsilon$ is polynomially small (in $N$, $\beta$, and the spectral gap $\Delta$), then the procedure remains close to equilibrium for all $s$, equilibrating along the path in a polynomial time.
It is important to emphasize that we do not conjecture that if there is a spectral gap, then the procedure at a fixed $s$ equilibrates in polynomial time for {\it any} starting trajectory, which clearly would not be true because at $s=1$ the procedure may not equilibrate rapidly, but only that the equilibration happens when we follow a particular path.

\subsection{Other Boundary Conditions}
Above, we have imposed periodic boundary conditions on $c_i$.  Other boundary conditions are possible.  For example, one can have open boundary conditions, where we evaluate
\begin{eqnarray}
\label{obc}
&&\langle \psi_1 | (\exp(-\beta H)) | \psi_1 \rangle
\\ \nonumber
&=&\sum_{\{c_i\}} \prod_{i=1}^{K-1} \langle \psi(c_{i+1}) | \exp(-\beta H/K) | \psi(c_{i}) \rangle,
\end{eqnarray}
where $\psi_1$ is the state $\sum_c |\psi(c) \rangle$.  One can also allow parameters to change with imaginary time; instead of taking $H=H_s$ for all times, we can choose weights
\be
\sum_{\{c_i\}} \prod_{i=1}^{K-1} \langle \psi(c_{i+1}) | \exp(-\beta H_{s(i)}/K) | \psi(c_{i}) \rangle,
\ee
where $s(i)$ depends upon $i$.  These different boundary conditions will be considered below.

\section{Counter-Examples}
We now describe several systems which serve as counter-examples to the conjecture above.  We will describe these systems not in terms of Ising spin degrees of freedom, but rather in terms of a particle moving in some potential.  That is, they will all be examples of single-particle quantum mechanics.  The annealing protocol for each system consists not of changing a transverse field but of changing some other parameters, and we consider paths in parameter space that are not just linear interpolations between two different Hamiltonians but are more general paths.
In the next section, we will describe some ``gadgets" to define Ising spin systems whose effective dynamics mimics that of the systems considered in this section, though again we consider more general paths in parameter space, allowing arbitrary tuning of Ising couplings and transverse fields.

We will describe the first system using a continuous space of states for the particle and use a discrete space of states for later examples.  In the case of a continuous space of states, we will briefly explain how to discretize the space of states.  In the first three examples, the number of such discrete
states will be proportional to a quantity that we write as $M$.
We will refer to coupling constants being ``polynomially large" if they are bounded by a polynomial in $M$, and we will consider a gap to be at most polynomially small if it is at least an inverse polynomial in $M$.  When we define the gadgets later, the needed value of $N$ (the number of Ising spins) will be a polynomial in $M$, so any quantities that are polynomials in $M$ will then be polynomials in $N$ also.  Similarly, if a quantity is exponentially small in a polynomial in $M$ (for example, tunneling between different winding number sectors in the first example), it will be exponentially small in a power of $N$, albeit possibly a power less than $1$; in general, we call such quantities exponentially small and do not worry about the power in the exponent.
In the fourth example, we will have exponentially many low energy states, and constructing the gadgets will be more complicated.

The reason for giving several counter-examples is that we will modify the algorithm in attempts to deal with some of the early counter-examples, and then we will provide further counter-examples which show that even the modified algorithm does not work in general.  Also, the later counter-examples show more severe effects and overcome some unsatisfactory features of the early counter-examples.

\subsection{First Example: Circle}
The first system (we will see later why we refer to this as a circle) is single particle quantum mechanics of a particle moving in two-dimensions with the Hamiltonian
\be
\label{Hcontin}
H=\frac{-1}{2m} \partial^2 + V(x,y),
\ee
where the derivative $\partial^2=\partial_x^2+\partial_y^2$.
We choose the potential
\be
V(x,y)=\mu (x^2+y^2)+g (x^2+y^2)^2 - h x.
\ee
We will fix $g=1$ throughout but we will change the coupling constants $\mu,m,h$ along the annealing path.
Note that at $h=0$, for $\mu<0$, the minimum of the potential is such that $x^2+y^2=-\mu/2g$.

The point of this example is not to show that the QMC procedure is unable to find the ground state, but rather to show that the QMC procedure must take an exponential time to equilibrate despite having only a polynomially small gap, thus contradicting the conjecture above.  We will develop other examples later where QMC does not find the ground state.

The continuous form is given for illustrative purposes only and we will work with a discretization instead.  Let us define quantities $R,a$, with $a<<1$ and $R>>1$.  The quantity $a$ will define a discretization scale at short distances and $R$ will define a maximum distance.  We will let $x,y$ become discrete, each being equal to an {\it integer} multiple of $a$, such that $|x|,|y| \leq R$.  The discretized Hamiltonian is
\begin{eqnarray}
H&= &-\frac{1}{2ma^2} \sum_{m,n} \Bigl( |x+a,y \rangle \langle x,y | + |x,y+a \rangle\langle x,y| \Bigr) + h.c. \nonumber \\
&& + \sum_{x,y}V(x,y)  |x,y \rangle \langle x,y |.
\end{eqnarray}
Note that $M \sim (R/a)^2$.

We follow the following annealing path.  We begin at $\mu=h=0$ and $m=1$.  Then, we make $\mu$ become negative, keeping $h=0$, until $\mu$ becomes equal to $-R^2$.  
For negative $\mu$, the minimum of the potential is at
\be
x^2+y^2=r_{min}=-\frac{\mu}{2},
\ee
so at $\mu=R^2/2$ the minimum of the potential is at $x^2+y^2=R^2/2$.  By choosing $a$ to be $1/\poly(R)$, we can make the effects of discretization negligible.  
Once $\mu$ becomes sufficiently negative, the particle is very narrowly confined near the minimum of the potential: transverse motion has a very high energy cost, with the second derivative of $V$ in the perpendicular direction being of order $\mu^2$.
The effective dynamics is that of a particle moving in a circle (hence the name of this example), with Hamiltonian
\be
\label{Hcirc}
H_{circ}=\frac{-1}{2mr_{min}^2} \partial_\theta^2+h R \cos(\theta),
\ee
where the angular variable $\theta$ is periodic with period $2\pi$ and where we have included the term $h$ in the potential because the next step involves increasing $h$.
The low-lying excited states at $h=0$ correspond to different angular momentum modes, with the lowest states having zero angular momentum and the first excited state being doubly degenerate, having angular momentum $\pm 1$.  The energy splitting is of order $1/r_{min}^2$, and hence the gap remains at most polynomially small up to this point.

We next increase $h$ from $0$ to $1$.  The gap is at most polynomially small in $R$ throughout this path.  At $h=1$, we can approximate the Hamiltonian by expanding near $\theta=0$:
\be
H_{circ} \approx \frac{-1}{2mR^2} \partial_\theta^2 -\frac{R}{2} \theta^2 + {\rm const.},
\ee
which is a harmonic oscillator Hamiltonian.  The wavefunction decays exponentially for $|\theta| \geq 1/R^{3/4}$, and so the particle is localized near $\theta=0$.

Finally, we increase $m$ to infinity.  Note that increasing $m$ corresponds to decreasing the corresponding term in the Hamiltonian, and is analogous to turning off the transverse field.  Note also that because of the discretization, the gap remains at most polynomially small as $m$ is decreased for appropriate choices of $R$, as follows.  At $m=\infty$, the eigenfunctions are localized on a single state $|x,y\rangle$, and so the eigenvalues are just the different values of the potential $V(x,y)$ at appropriate discrete values of $x,y$.  One could pick $R$ badly so that $V(x,y)$ has a pair of degenerate minima in the discrete case, but it is easy to avoid this.
At the end of this process, the particle has found its ground state.

Now, let us analyze what happens with the QMC annealing procedure.
The worldline of the particle is some closed path in imaginary time.
Once $\mu$ is sufficiently negative, 
the distribution of trajectories has very small probability to include any point with radius $x^2+y^2$ which differs much from $r_{min}^2/2$ and
we expect that it takes an exponential time to transition from one winding number sector to another.

We can estimate the {\it equilibrium} winding number of a trajectory for Hamiltonian (\ref{Hcirc}) at $h=0$ and arbitrary $r_{min}^2$ as follows.   The contribution to the partition function of the trajectories with winding number $n$ is proportional to the Green's function $G(0,2\pi n;\beta)$, where this denotes the Green's function for a particle moving on the line (i.e., the universal cover of the circle) to move from $0$ to $2\pi n$ in imaginary time $\beta$.  This contribution is proportional to
$\exp(-mr_{min}^2(2\pi n)^2/2\beta))$, which decays exponentially for $n \gtrapprox \sqrt{\beta/m}/r_{min}$.  Hence, a typical trajectory will have a winding number $n \sim \sqrt{\beta}$.

Note that this equilibrium value depends upon $r_{min}$, so
the system must necessarily fall out of equilibrium, assuming that $\mu$ changes on a time scale faster than the exponential time required to change between winding number sectors. This already contradicts the conjecture that QMC will equilibrate.
A similar effect would happen if we considered an annealing protocol, such that we equilibrated the system at fixed, large, negative $\mu$ (taking the exponential time necessary to equilibrate between winding number sectors) and then increased $m$.

Now we return to the annealing protocol discussed above, where $h$ increases.  In this case, if the system is stuck with $x^2+y^2 \approx R^2/2$ so that Hamiltonian (\ref{Hcirc}) applies, what we find is that the trajectory spends most of its time near $\theta=0$, and then spends the rest of the time winding around.
Suppose, for example, the system is stuck in a sector with winding number $n \neq 0$.  We can calculate an instanton trajectory (i.e., find a minimum action trajectory with the given winding number).  The minimum action solution can be described by a particle which spends a long time near $\theta=0$ at a slow speed, then accelerates and rapidly winds around the circle, then again spends a long time near $\theta=0$, and again rapidly winds around the circle, and so on, until it winds a total of $n$ times.  We can give a quick estimate of the time it spends winding as follows.
Suppose out of a total imaginary time $\beta$, the system spends time $\beta-\tau$ close to $\theta=0$ and spends time $\tau<<\beta$ rapidly winding $n$ times around the circle.   We can estimate the optimum time $\tau$ by estimating the action for the optimal trajectory with given $\tau$ and minimizing over $\tau$.  The action to wind $n$ times in time $\tau$ is of order $mR^2n^2/\tau+h \tau R$ where the first term comes from the kinetic energy and the second term comes from the potential energy. This is minimzed at $\tau$ of order $n \sqrt{Rm/h}$.

If we fix $R$ and take $\beta$ large, since $n\sim \sqrt{\beta}$we have $\tau \sim \sqrt{\beta}$ so that the trajectory spends {\it most} of its time near $\theta=0$.  Thus, in some sense the QMC procedure does
succeed in finding the minimum as the trajectory spends most of the time in the correct place, assuming that $\beta$ is sufficiently large.

It is interesting to analyze what happens as $m$ is increased to infinity, starting at $h=1$.  The time $\tau \sim n \sqrt{Rm/h}$ estimated above increases, eventually becoming of order $\beta$.  However, for sufficiently large $m$, the trajectories stop being localized near $x^2+y^2=R^2/2$ and instead the trajectory gets localized at smaller $x^2+y^2$.  We can understand this as a balance of two terms (we will consider $h=0$ for simplicity): if the trajectory is localized at a given distance $r$, so that $x^2+y^2 \approx r$, the action to wind $n$ times in time $\beta$ is of order $mr^2n^2/\beta+V(r)$, and for large $m$ this is minimized at small $r$, so at large enough $m$ the trajectories start to have non-negligible probability to have $x=y=0$.
Thus, at large $m$, the system {\it is} able to change its winding number to zero in a time which is not exponentially long because the trajectories move to smaller $r$.

Eventually, at very large $m$, the trajectory becomes close to constant in imaginary time (each $c_i$ is close to all other $c_j$).  At this point, the QMC dynamics becomes similar to a classical Monte Carlo dynamics as the trajectory is almost determined by its value at a given time slice and so we can just study a classical Monte Carlo procedure with weight $\exp(-\beta V(x,y))$.  This is unsurprising, as at large $m$, the dynamics becomes more classical.  At large $\beta$, this classical Monte Carlo procedure is similar to a greedy algorithm, as with high probability the particle only moves to lower potential states.

The particular potential we have chosen has the property that it has only one local minimum.  As a result, the classical Monte Carlo dynamics will not get trapped and eventually the system will equilibrate at large enough $m$, being stuck just at the minimum.  We can modify the example by changing the potential near $x=y=0$, adding an additional minimum there, to trap the large $m$ dynamics to construct an example which prevents this equilibration.

To summarize: this example uses winding number as a topological invariant of trajectories to construct an annealing protocol for which the QMC does not relax rapidly.
However, for various reasons, this example is not completely satisfactory as a counter-example to the idea that QMC will succeed in finding the minimum when the annealing algorithm does.
One such reason is that if $\beta$ is sufficiently large, the QMC does produce trajectories which spend {\it most} of their time near the desired minimum at the point in the annealing protocol when $m=h=1$. 
A second reason is that the QMC algorithm has some probability of being in the trivial sector with winding number $n=0$, where it can more readily equilibrate at intermediate values of $m$ (i.e., small enough $m$ that the winding number sector is still fixed but large enough that the equilibrium distribution is dominated by the sector $n=0$) and this probability of being in the trivial sector is only polynomially small.

One interesting attempt to modify the QMC procedure to help equilibration, or at least to ameliorate the effects of being stuck in a sector with non-zero winding number, is to change $\beta$ during the procedure.  Let us again return to analyzing the protocol where $h$ is kept at $0$, but $m$ is increased.  In this case, we could allow $\beta$ to increase also, and if $\beta/m$ are in the right ratio, the system will remain in equilibrium at given $\beta$.  We will discuss ideas like this again in later examples.

\subsection{Second Example: Bouquet of Circles, ``Too Long a Word"}
The first example was based on a case where the particle was confined (up to exponentially small corrections) to a circle.  The fundamental group of the circle is $\mathbb{Z}$ and is abelian.  In this case, the equilibrium state had a winding number proportional to the square-root of $\beta$.
In the next example, we consider a case where the fundamental group is a non-abelian group.  We consider a system where the particle is confined to a space which is a bouquet of circles.  A bouquet of circles consists of several circles glued together at one point.  The fundamental group of a bouquet of $n$ circles is the free group on $n$ generators.  For simplicity, we consider a bouquet of $2$ circles.

This example makes the effects in the previous section more severe, especially in the large $\beta$ case.   This example also introduces ideas used in later examples.

The different topological sectors are described by words in the free group.
For example, if we have two generators, called $a,b$, then a possible word is $aba^{-1}  b$.  Words such as $abaa^{-1}a^{-1} b b$ can be reduced by cancelling the successive appearance of a generator ($a$) and its inverse ($a^{-1}$), and in fact $abaa^{-1}a^{-1}b=aba^{-1}b$, and the two words describe the same topological sector.  Further, we can cyclically reduce a word (cancel a generator at the start of the word against its inverse at the end) and $b^{-1}aba^{-1}bb$ describes the same sector as $aba^{-1}b$.

For a given topological sector, we can ask for the length of the short possible word that describes that sector.  This is the length of the cyclically reduced word.  Thus, for $aba^{-1}b$, the length is $4$.

A possible system with such a fundamental group has $2M-1$ basis states.  There are two sequences of states, labelled $|i,a\rangle$ and $|i,b\rangle$, where $i$ is an integer in $1,...,M-1$.  Additionally there is one other state labelled $|0\rangle$.
We will construct a Hamiltonian whose effective low energy dynamics is given by $H_{bouquet}$, defined by
\begin{eqnarray}
H_{bouquet} &=& -\sum_{i=1}^{M-2}\sum_{x \in \{a,b\}}\Bigl( |i+1,x\rangle \langle i,x |+h.c. \Bigr)
 \nonumber\\
&&+ 2 \sum_{i=1}^{M-1} \sum_{x \in \{a,b\}} |i,x\rangle \langle i,x| \\ \nonumber
&&- \sum_{x \in \{a,b\}} \Bigl( |1,x\rangle \langle 0 | + h.c. \Bigr) \\ \nonumber
&&- \sum_{x \in \{a,b\}} \Bigl( |M-1,x\rangle \langle 0 | + h.c. \Bigr) \\ \nonumber
&& + 4 |0\rangle\langle 0|.
\end{eqnarray}
The diagonal terms in the Hamiltonian are chosen so that the ground state of $H$ is an equal amplitude superposition of all states.

See Fig.~\ref{bouquet}.  The Hamiltonian is equal to the graph Laplacian on the graph shown, where we have shown the case $M=8$.  The graph Laplacian  has an off-diagonal element equal to $-1$ between any two vertices connected by an edge, and has diagonal elements equal to the degree of a given vertex (so the vertex at the middle of the figure has degree $4$ and all the others have degree $2$).
\begin{figure}
\includegraphics[width=1in]{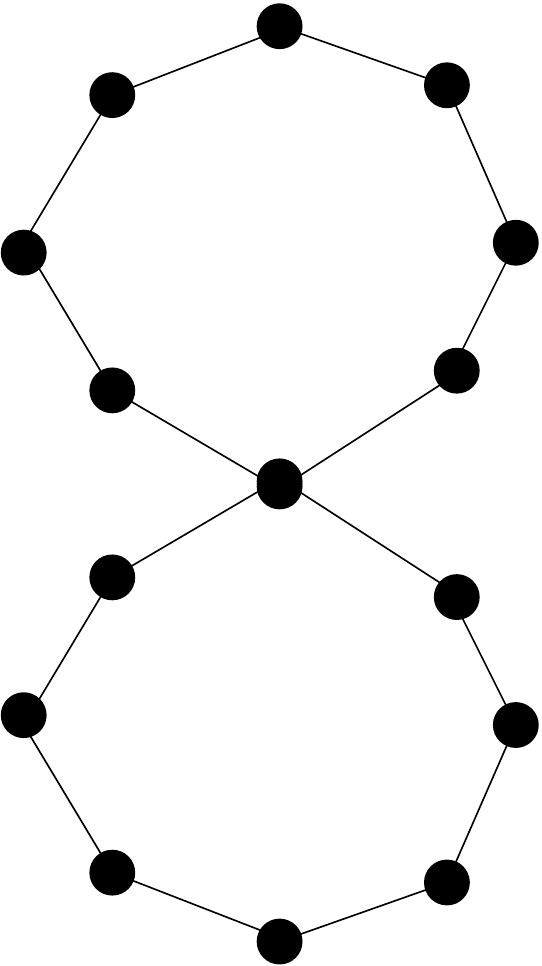}
\caption{Illustration of a graph corresponding to a bouquet of two circles for $M=8$.}
\label{bouquet}
\end{figure}


The Hamiltonian $H_{bouquet}$ will be an effective Hamiltonian for some other Hamiltonian with a larger number of states, similar to how$H_{circ}$ was an effective Hamiltonian previously.  We add additional states to the system and follow some annealing protocol so that at some point in the protocol $H_{bouquet}$ becomes a description of the effective dynamics {\it and} so that the typical trajectory created by the QMC algorithm is in a topological sector whose shortest word length is proportional to $\beta$.
We now sketch one way to do this, but the reader can certainly imagine many possible ways.
We have drawn the bouquet of circles in the plane.  We can imagine that the particle moves through the plane similarly to the previous case, and that it is some potential $V(x,y)$ that confines it to the two circles (and that also produces the appropriate diagonal and off-diagonal terms in $H_{bouquet})$).
We can imagine that we follow an annealing protocol so that initially the particle is able to move throughout some region of the plane, and that we change some parameter so that eventually the particle gets confined to the given bouquet.

In order to exponentially suppress the motion away from the two circles, we might want to take $M$ polynomially large.  Also, later we like to take $M$ polynomially large to localize certain states as we turn on a potential term $h$ in Eq.~(\ref{Hbpot}) below.  Otherwise, the particular value of $M$ is not that important, though we do need $M\geq 3$ to define the direction of winding around a circle.

Having quenched to this Hamiltonian $H_{bouquet}$, we now imagine an annealing protocol for a Hamiltonian of the form
\be
\label{Hbpot}
H=\frac{1}{m} H_{bouquet}-h |0\rangle\langle 0|.
\ee
We pick $h \geq 0$ and the term $h$ is added to produce a minimum in the potential.  The particular choice of the minimum being state $|0\rangle$ as opposed to some other state is unimportant.
Changing $m$ plays a similar role to before, and at large $m$ the equilibrium trajectory is close to constant.  
The annealing protocol from this point is: start at $m=1$ and $h=0$.  Then, increase $h$ to $1$.  Finally, increase $m$ to $\infty$.

As in the previous case, if we change $m$ but keep $h$ fixed at zero, the system must fall out of equilibrium if it is unable to transition between different topological sectors.
To analyze this, in equilibrium, at $h=0$, the length of the short possible reduced word for a typical winding number sector is as claimed above,
\be
{\rm const.} \cdot \beta/(m M^2)
\ee
This follows from the fact that a trajectory in imaginary time describes a closed random walk on this bouquet of circles.  It takes time of order $mM^2$ for the random walk to go once around a circle, and so the trajectory corresponds to a random word of length $\beta/M^2$.  However,
 for a random word of given length on the free group with two generators, the length of the corresponding cyclically reduced word is typically only a constant fraction smaller than the given random word.
So, as $m$ changes, this word length changes.

We now analyze what happens as $h$ increases at fixed $m=1$.  In this case if we take $h$ large, we can make the ground state localized near $|0\rangle$. Indeed, if we take $h$ of order unity, then the ground state is exponentially localized near $|0\rangle$, and by taking $M$ large we can exponentially suppress states of distance $\sim M/2$ from $|0\rangle$.
However, suppose we have a trajectory stuck in a topological sector with word length of order $\beta$.  Then, we can perform a similar instanton analysis as before.  We do the instanton analysis by going to a continuum limit and finding a minimum action trajectory.  We study this trajectory on the universal cover of the bouquet of circles (this cover is a tree graph), where the the trajectory travels a distance of order the word length (i.e., of order $\beta$) in time $\beta$.  So, the trajectory has an action proportional to $\beta$.  At the point that $m=h=1$ in the annealing protocol, the fraction of time that this trajectory spends near $|0\rangle$ is $\beta$ independent, differing from what we found when the target space was just a single circle where in the large $\beta$ limit the trajectory spends most of its time near $|0\rangle$.  This is a result of the word length being of order $\beta$ now rather than $\sqrt{\beta}$. as previously

To summarize, again we find problems equilibrating, and again it occurs because the system is typically stuck in a nontrivial topological sector.
We have referred to this example as ``too long a word" for this reason.

\subsection{Third Example: Bouquet of Circles, ``Too Short a Word"}
\label{ThirdOne}
Since the previous examples were both based on a situation in which the system is stuck in a nontrivial topological sector, but the minimum action sector is the trivial topological sector, one might imagine trying to modify the QMC algorithm to cause the dynamics to be stuck in the trivial sector, or at least stuck in a sector with a shorter word length than typical.
For example, we could follow an annealing protocol in which we change both $\beta$ and $s$.  One possibility would be to {\it increase} $\beta$ during the annealing protocol.
We could either increase $\beta$ while keeping the number of time slices constant (in which case the change in $\beta$ leads to a change in the statistical weights for a trajectory), or we could also change the number of time slices.  For example, one possible way to increase the number of time slices by one is to replace a trajectory $c_1,...,c_K$ by a trajectory $c_1,...,c_K,c_{K+1}$, setting $c_{K+1}=c_K$.
The goal of increasing $\beta$ in this way would be to make the system be closer to the trivial sector for the given $\beta$; that is, consider the first example of a circle.  The winding number is proportional to $\beta$.  If we equilibrate the winding number at a given $\beta$ and then double $\beta$, the winding number is now smaller than expected for the given $\beta$.

An alternate approach would be to combine this increasing in $\beta$ with a dependence of $s$ upon the time slice; we do not discuss this further as this example will be hard for such a case too.

We now consider an example for which such an approach would not work.
Consider a system with the same states as above for the bouquet of circles example, and one additional state called $|r\rangle$.
Let the Hamiltonian be
\begin{eqnarray}
H&=&mH_{bouquet} - h |0\rangle\langle 0| \\ \nonumber
&&- t |0\rangle \langle r| + h.c. \\ \nonumber
&&+E|r\rangle\langle r|.
\end{eqnarray}
That is, we have added a tunneling term $t$ connecting the state $|0\rangle$ to the state $|r\rangle$ and also added a potential term $E$ for state $|r\rangle$.
Note that these states (the bouquet and added state $|r \rangle$) are the only states we consider; this differs from the previous example where we considered a system with a larger number of basis states, and changed some parameter to confine the particles motion to
the bouquet, quenching into a topologically nontrivial sector.

Note that the ground state of $H_{bouquet}$ has energy $0$, and the first excited state of $H_{bouquet}$ has at least energy $c/M^2$, for some positive constant $c$.  It will be important in what follows to consider also the spectrum of the Laplacian on the universal cover of the bouquet of circles.  This cover is a tree.  This tree can be constructed as follows: start with a tree $T$ such that all nodes have degree $4$; that is, the root has $4$ daughters and all other nodes have $3$ daughters.  Take this tree and insert $M-1$ additional vertices in the middle of each edge, to construct a new tree $T'$; that is, replace an edge between two nodes $v,w$ by an edge from $v$ to $v_1$ then from $v_1$ to $v_2$ and so on, up until $v_{M-2}$ to $v_{M-1}$, and then an edge from $v_{M-1}$ to $w$.
The spectrum of the Laplacian on $T$ is in $[4-2\sqrt{3},\infty)$.  The exact value $4-2\sqrt{3}$ is not that important; what is important is that this value is greater than $0$.  
Similarly, the spectrum of the Laplacian on $T'$ is at least $[c'/M^2,\infty)$ for some positive constant $c'$.

We follow this annealing protocol: start at $h=t=0$, and with $E$ being large and negative, so that the ground state at the start is $|r\rangle$.
Increase $t$ from $0$ to $c/2M^2$.
Then, increase $E$ until $E$ is  ${\rm min}(c,c')/2M^2>0$.
Then, decrease $t$ to $0$.
Then, increase $h$ to $1$ and finally increase $m$ to $\infty$; this last stage of the annealing protocol is the same as in the previous example.

Once $E$ reaches its maximum value, the ground state of the quantum Hamiltonian is some superposition of $|r\rangle$ and some state on the bouquet.  Since $t$ is smaller than the energy of the first excited state of the bouquet, the state on the bouquet has most of its amplitude on the ground state of the bouquet.  This process can be understood as an avoided crossing: the energy of $|r\rangle$ crosses zero (the energy of the ground state on the bouquet), but because of the non-zero $t$ the crossing is an avoided crossing.

As $t$ is decreased, the amplitude of the ground state on $|r\rangle$ decreases, until that amplitude is zero once $t$ reaches zero.  Finally as $h$ is increased, then the amplitude of the ground state on $|0\rangle$ increases, until at the end of the protocol the ground state is exactly $|0\rangle$, as in the previous example.  Note that throughout the spectral gap is only polynomially small in $M$.

Now we consider what happens for the QMC protocol.  We number our configurations in the natural way: we write $\psi(r)=|r\rangle$ and write $\psi(k,x)=|k,x\rangle$.  Initially, we have $c_i=r$ for all $i$.  However, once $t$ becomes non-zero, the trajectory starts to spend time on the bouquet.
Suppose at some pair of time slices $i,j$, we have $c_i=c_j=r$, but for all times $k \in \{i+1,i+2,...,j-1\}$ we have $c_k \neq r$.  Then, the sequence of configurations $c_{i+1},...,c_{j-1}$ is some sequence that starts and ends at $|0\rangle$ and forms a {\it topologically trivial} path.  
To understand this, let us refer to a sequence of time slices such as $i+1,...,j-1$ such that the particle is on the bouquet and such that $c_i=c_j=r$ as an ``interval".  The length of an interval can increase or decrease under the dynamics but the topologically sector cannot change.  Initially, there are no intervals.  When a new interval is created, it is created as a single time slice, containing only one configuration, $0$.  The trajectory on this interval is topologically trivial.

As a result of this constraint on the topology of the trajectory during the intervals, the QMC dynamics is unable to distinguish between the given Hamiltonian, and a Hamiltonian where we have coupled the state $|r\rangle$ to the {\it universal cover} of the bouquet.  However, the energy $E$ is at all times less than the bottom of the spectrum of the Laplacian on the cover of the bouquet, since we have chosen $E < c'/M^2$.  So there is no avoided crossing simply because even at $t=0$ there would be no crossing.  When $t$ is reduced back to $0$, the ground state of the Hamiltonian coupling $|r\rangle$ to the universal cover of the bouquet is the state $|r \rangle$.  Hence, the QMC procedure produces the states $|r\rangle$ and does {\it not} find the correct ground state of the Hamiltonian at the end of the annealing protocol.  In fact, if the QMC is perfectly equilibrated within the trivial topological sector, then we find that the QMC has zero probability of finding the ground state.

This situation is then much worse than the previous examples.
The failure of QMC can be understood in a different fashion.  Suppose we have $t=0$.  Then, the partition function is a sum of two different quantities, one being the partition function of the bouquet and one being the contribution $\exp(-\beta E)$ from the state $|r \rangle$.  For large $\beta>>M^2$, the partition function of the bouquet approaches $1$, as the bouquet has a unique ground state with energy $0$.  However, this partition function $1$ is a sum of contributions from exponentially many different topological sectors.  Any given topological sector has a contribution which is exponentially suppressed in $\beta/M^2$.  Hence, an algorithm that is unable to equilibrate between sectors greatly underestimates the contribution of the bouquet to the partition function.

\subsection{Fourth Example: Bouquet of Circles, Open Path in Imaginary Time}
\label{FourthOne}
All of our examples so far have been based on a nontrivial fundamental group.  A natural question is whether we can resolve these problems with QMC by using open boundary conditions instead.  To motivate this, if we consider classifying closed paths in some space, then the fundamental group $\pi_1$ enters, but if we classify open paths in some space (i.e., continuous functions from an interval $[0,1]$ to some space, with no requirement that $0$ and $1$ be mapped to the same point), then the classification of such open paths is the same as $\pi_0$: if the space is path connected, then any two such open paths can be deformed into each other.

In this example, we show that such a QMC algorithm still does not necessarily work.  The example builds off our previous example. We still have a bouquet of circles, but now in addition to adding the state $|r\rangle$ as in the previous example, we also add another set of $N_G$ different states, labelled $|1,G\rangle,...,|N_G,G\rangle$, and define some expander graph whose vertices correspond to the states $|i,G \rangle$.
We let the Hamiltonian be:
\begin{eqnarray}
H&=&mH_{bouquet} - h |0\rangle\langle 0| \\ \nonumber
&&- t |0\rangle \langle r| + h.c. \\ \nonumber
&&+E|r\rangle\langle r| \\ \nonumber
&& + L_{expander}+V P_{expander} \\ \nonumber
&&-t' |r\rangle\langle 1,G|+h.c.
\end{eqnarray}
where $L_{expander}$ is the graph Laplacian on the expander, and $P_{expander}$ is a diagonal matrix equal to $1$ for states on the expander and $0$ otherwise.

First we analyze the properties of the part of the Hamiltonian that acts on $|r\rangle$ and on the expander:
\begin{eqnarray}
\label{expHam}
&&E|r\rangle\langle r| \\ \nonumber
&& + L_{expander}+V P_{expander} \\ \nonumber
&&+t' |r\rangle\langle 1,G|+h.c.
\end{eqnarray}
The Hamiltonian $L_{expander}+V P_{expander}$ has ground state of energy $V$, and then a gap to the rest of the spectrum.  For $E<V$ and $t=0$, the ground state of the Hamiltonian (\ref{expHam}) is $|r\rangle$ with energy $E$.  For $E<V$ and small $t$, the ground state is a superposition of some state on $|r\rangle$ and some state on the expander.  This state on the expander has its largest amplitude on $|1,G\rangle$, with the next highest amplitude on the first neighbors of $|1,G\rangle$, and the amplitude decreasing away as we consider further neighbors from $|1,G\rangle$.
Note that for $E<V$ and $t<<|E-V|$, the ground state has almost all of its {\it probability} on $|r\rangle$, and has only a probability of order $|t^2|/|E-V|$ on the expander.  However, we can choose $E$ close to $V$ so that the following happens: almost all of the {\it amplitude} of the state is on the expander.  That is, if we pick a site with probability proportional to the {\it amplitude} of the ground state wavefunction, then the result is very likely to be on the expander.  By taking $N_G$ large, we can make the probability of the ground state wavefunction strongly concentrated on $|r\rangle$ but the amplitude strongly concentrated on the expander.
We write $E_0(E,V,t')$ to denote the ground state energy of Hamiltonian (\ref{expHam}).

We begin the annealing protocol with $t=t'=0$, and we choose $E$ to be negative with $|E|>>1$ so that the initial state is highly concentrated (in both probability and amplitude) on $|r\rangle$.  We then make $t'$ slightly non-zero, and adjust $V$ so that the above regime holds, with the probability concentrated on $|r\rangle$ and the amplitude concentrated on the expander.  From this point on in the annealing protocol, we maintain the same difference $E-V$, adjusting $V$ to keep this difference constant whenever $E$ is adjusted.  We also keep the same $t'$.  We then follow a very similar annealing protocol to the previous example:
increase $t$ from $0$ to $c/2M^2$.
Then, increase $E$ and $V$ until $E_0(E,V,t')$ is  ${\rm min}(c,c')/2M^2$, while keeping $E-V$ constant.
Then, decrease $t$ to $0$.
Then, increase $h$ to $1$ and finally increase $m$ to $\infty$.

We can choose the difference $E-V$ so that even when we increase $t$ to $c/M^2$ and $E_0(E,V,t')$ to ${\rm min}(c,c')/2M^2$ and $t$, the ground state wavefunction has most of its amplitude on the expander.
With open boundary conditions, using Eq.~(\ref{obc}) for statistical weights, the variables $c_1$ and $c_K$ are correlated.  However, in the limit of $\beta>>\Delta$, in equilibrium the joint probability distribution approximately factorizes:
\be
P(c_1,c_K) \approx \psi_0(c_1) \psi_0(c_K),
\ee
where $\psi_0(c)$ is the amplitude of the ground state wavefunction, normalized so that $\sum_c \psi_0(c)=1$.  That is, the probability distribution of $c_1$ and $c_K$ are governed by the amplitudes of the ground state, and are very likely to be on the expander graph.

If $c_1$ and $c_K$ stay on the expander graph throughout the QMC simulation, then the topological sector cannot change, and we find the same effect as in the previous example that the QMC algorithm will be very unlikely to find the correct ground state.
So, we must ask for the probability that $c_1$ or $c_K$ does leave the expander graph.  By taking $N_G$ exponentially large, we can make this probability exponentially small.

This example shows that even open boundary conditions need not solve the problem, because we can define a Hamiltonian so that $c_1,c_K$ are ``pinned points".  That is, they are fixed to be on the expander graph, preventing a change in topological sector.
Unlike the previous examples, we need to use an exponentially large number of states, and so it will take a little more care to define gadgets for this example in the next section.

\section{Gadgets: From The Tranverse Field Ising Model to More General Hamiltonians}
We now describe how to construct transverse field Ising model Hamiltonians whose effective low energy dynamics realizes the four examples considered above.
We first consider the first three examples.  For these examples, we need to construct an effective Hamiltonian with a number of states $\sim M$ that scales as some polynomial in $N$.

The transverse field Ising systems that we consider will not necessarily be planar, and we will allow arbitrary dependence of the transverse field and Ising couplings along the annealing protocol.  We leave it as an open question whether one can construct gadgets using planar Hamiltonians where only a single parameter, the strength of the transverse field, is tuned.

Consider first a Hamiltonian
\be
H_{global}=J_{AF} \sum_{i,j} S^z_i S^z_j+h \sum_i S^z_i,
\ee
where the sum is over all $i,j$.
In an eigenstate with a total of $N_\uparrow$ of the spins up, and $N-N_\uparrow$ spins down,
the energy is $J_{AF} (2 N_\uparrow-N)^2+h(2 N_\uparrow-N)$.  By tuning $J_{AF},h$, we can arrange for this to have a minimum at any desired value of $N_\uparrow$, and with a gap of order unity to any states with a different value of $N_\uparrow$.

However, if $N_\uparrow\neq 0,N$, there are many different states with the given value of $N_\uparrow$.  We will construct an effective Hamiltonian
in this space of states that realizes the desired examples previously.  We first describe how to construct a Hamiltonian whose effective low energy dynamics realizes a hopping Hamiltonian on a circle:
\be
\label{Hring}
H=-t \sum_a |a\rangle\langle a+1|+h.c.+...,
\ee
where $a$ is periodic with period $M$, and where $...$ represent terms $t_k \sum_a |a\rangle \langle a+k|+h.c.$ for $k>1$, with $t_k$ decaying rapidly in $k$.  
Set $N=M$.  Label different sites by $i$, with $i$ being periodic with period $N$.  Pick an integer $R$. Take the Hamiltonian
\be
H=H_{global}+J' \sum_{|i-j|\leq R-1} S^z_i S^z_j-B \sum_i S^x_i.
\ee
Tune $J_{AF},h$ so that the ground state of $H_{global}$ has $N_\uparrow=R$.
First we consider the case $B=0$.  For $J'<0$, we have a short range ferromagnetic interaction.  For $|J'|<<1$, we find that the low energy states
consist of states with $N_\uparrow=1$ and with all the up spins next to each other.  That is, a state of the form $\downarrow \downarrow ... \downarrow \uparrow \uparrow.... \uparrow \downarrow \downarrow ... \downarrow$, where there is exactly one sequence of up spins with length $R$.
There are $N=M$ such different states.

Taking $B$ small compared to $J'$, we can treat $B$ in perturbation theory, and at second order, the effect is to allow the sequence of up spins to move either one to the right or one to the left.  That is, we flip a spin on one side of the sequence from up to down, shortening the sequence on one side, while flipping a down spin just past the end of the sequence on the other side to up, lengthening the sequence on that side.

Now let us explain why we introduce the parameter $R$.  At $R=1$, at second order in perturbation theory the sequence of up spins can move anywhere: there is no distinction between different sides of the sequence and two spin flips can connect any two states with exactly one up spin.  For $R=2$, the second order perturbation theory result gives us the desired effect, but at fourth order in perturbation theory, we can move the sequence anywhere.   For $R=3$, the sixth order perturbation theory allows us the sequence to move anywhere and there is a term at fourth order in perturbation theory that contributes to $t_2$, moving the sequence by two.  However, we can take $R$ to be a polynomial in $N$, with a power less than $1$, and take $B$ polynomially small, and then $t_k$ decays as an inverse polynomial of $N$ raised to the $k$-th power for $k<R$ and is negligible for $k \geq R$.

This lets us realize a Hamiltonian of form (\ref{Hring}).  We can add additional magnetic fields by a term $\sum_i h_i S^z_i$, allowing us to realize a Hamiltonian 
\be
\label{eff1}
H=-t \sum_a |a\rangle\langle a+1|+h.c. + \sum_a V(a) |a\rangle\langle a|+...,
\ee
where $V(a)=\sum_{i=a}^{a+r-1} h_i$, giving a linear map from $h_i$ to $V(a)$, and where again the
$...$ represent terms $t_k \sum_a |a\rangle \langle a+k|+h.c.$ for $k>1$.  This linear map is not invertible, so not all $V(a)$ are possible.  However, we are able in this way to approximate a slowly varying potential $V(a)$.  This allows us to approximate continuum equations such as Eq.~(\ref{Hcontin}) by some discrete approximation.  It is a slightly different discrete approximation than before, as we have hopping $t_k$ beyond the first neighbor.  However, it still approximates a continuum equation.

Using these continuum equations as a building block, we can obtain any of the first three examples.  Note that we have given gadgets to realize a discrete approximation of a one-dimensional continuum equation, while these examples require a two-dimensional continuum equation.  This is a simple modification though.

To realize the fourth example, we need to overcome the fact that $N_G$ should be exponentially large.  We use the above gadgets to construct a Hamiltonian whose low energy dynamics has the states on the bouquet of circles, the state $|r\rangle$, and the state $|1,G\rangle$ on the expander graph.  To realize additional states on the expander graph, we add an additional $N'$ number of Ising spins.  We add a magnetic field to these spins which make them prefer to be down, but we add an additional ferromagnetic interaction between each of these spins and the spins used to construct the state $|1,G\rangle$.
This ferromagnetic interaction is chosen so that if the particle is not in the state $|1,G\rangle$, then these additional $N'$ spins will prefer to be down, but otherwise they have equal energy to be up or down.  We identify the $2^{N'}$ states where the particle is in $|1,G\rangle$ and these additional $N'$ spins are in arbitrary states with the states $|1,G\rangle,...,|N_G,G,\rangle$, where $N_G=2^{N'}$.

\section{Further Obstructions With A Trivial Fundamental Group}
\label{notopo}
We have noted two possible obstructions to equilibration, one based on a nontrivial $\pi_0$ and one based on a nontrivial $\pi_1$.
A natural question is whether these are the only obstructions.  This question is not completely well-defined, since it is not completely clear in general for which space we must compute $\pi_0$ and $\pi_1$; we have only described it as the space in which the wavefunction has non-negligible amplitude in some imprecise way.  However, in this section we explore this question and identify other obstructions.
To make it simpler to define the space in this section, we will imagine that the Hamiltonian is simply the Laplacian on some space.  Then, the ground state wave function has the same amplitude everywhere.  This space could be a continuous or discrete space.  We will construct examples where this space is simply connected and the Laplacian has only a polynomially small spectral gap (polynomially small in the volume of the given space) so that {\it particles} diffuse in a polynomial time but for which diffusion of {\it paths} is exponentially slow.

We give three different examples.  The first example shows slow equilibration starting from certain initial trajectories, but unfortunately these conditions are unlikely to occur in the QMC annealing since they have very small statistical weight.  The second example shows slow equilibration for initial trajectories with large statistical weight.  This slow equilibration is due to two effects.
Although the space of trajectories is connected, so that there is a sequence of trajectories connecting any given pair, this sequence is very long {\it and} this sequence requires going through trajectories with much smaller statistical weight.
This second point is analogous to what happens in Eq.~(\ref{doublewell}) where there is a path connecting the two wells but the path has low amplitude.  What we would really like is a case where we can connect any trajectory to any other by a sequence of trajectories with only slightly smaller statistical weight, but for which diffusion is still slow because it requires a very {\it long} sequence of trajectories to get from one to the other.  In the third example, we provide this, and after giving this example we then modify it to construct a problem on which the QMC annealing protocol will likely fail.

\subsection{First Example}

In a first example, consider a surface embedded in three-dimensions that looks like a dumbbell, with a narrow neck near the middle.  This surface is intended to have the topology of a sphere, and so $\pi_1$ is trivial.
As the neck pinches off, the Cheeger constant\cite{CC} goes to zero, but we will not need to take the neck that narrow.  Let us suppose that the neck has width $\sim 1$, while two halves of the dumbbell each have linear size $\sim L$ and area $\sim L^2$.  Then, the Cheeger constant is $\sim 1/L^2$.  We will take $L$ only polynomially large.
Now, imagine a path that takes an imaginary time $\beta$ to wind around the neck $~\beta$ times.  This path can be shrunk to a point, by pulling it off the neck, but because of the large number of windings around the neck, changing the number of times it winds around the neck is exponentially suppressed in $L$.

While this example allows us to have a polynomially small gap for a particle, but have an initial path that takes an exponential time to relax to equilibrium, this is slightly unsatisfactory, because the initial path is exponentially unlikely to occur in equilibrium, having a very small statistical weight.

\subsection{Second Example: Group Presentation and Sequence of Trajectories}
We fix this problem in the second example, which is based on group theory.
Recall the concept of a presentation of a group.  This is defined by certain generators, $g_1,...$ and certain relations, $r_1,...$, where a relation specifies that a certain product of the generators is equal to the identity.  For any presentation of the group, we can define a $2$-complex, called the ``presentation complex", whose fundamental group is the group specified by that presentation\cite{Hatcher}.  This complex is constructed as follows.  There is one $0$-cell.  For each generator, we attach a $1$-cell, giving a bouquet of circles, and then for every relation we attach a $2$-cell whose boundary is attached to the $1$-cells corresponding to the generators in that relation.  While this construction gives a complex, we can embed the complex without self-intersection in five dimensions and then construct a four manifold with the same fundamental group.

Note that previously, when we considered the bouquet of circles, we had some number $M$ of discrete states on each circle.  We can imagine doing something similar here, giving a finer subdivision of the complex.  Then a trajectory in the QMC simulation is some closed path on the $1$-skeleton of this subdivided complex (we use the term ``path" rather than ``trajectory" here for consistency with terminology in graph theory and to emphasize that the path is on some complex that we have defined from a group presentation, but we use ``trajectory" later when giving an analysis of a QMC algorithm for a system including both such a complex and some additional states).

Moving the path by local updates corresponds to using $2$-cells to change this path.  However, we can avoid explicitly doing
 this subdivision and instead just speak of a path as a word, where a word is a sequence of generators.  To do this, note that we can deform any path in the subdivided complex onto the $1$-skeleton of the original complex, by deforming it within each $2$-cell of the original complex.  We can arbitrarily decide some way to do this for each $2$-cell of the original complex.  Then, using a sequence of local updates changes the undeformed path, and may potentially change the deformed path.
This change in the deformed path corresponds simply to some relation (or its inverse).
So, for simplicity, we speak of a path as a word in the generators.  Local updates correspond to using a relation, or its inverse, or cancelling a generator against its inverse (so $g_a g_a^{-1}$ can be cancelled) or the inverse of this cancellation.

It is known\cite{Adjan,Rabin} that it is undecidable whether or not a finite group presentation is trivial.  Note that a group presentation will be trivial if and only if for every generator $g_a$, we can take a path that winds exactly once around the $1$-cell corresponding to that generator and deform that path to the identity.
If there were some sufficiently slowly growing bound (indeed a bound by any computable function) on how long this sequence of paths might be, then we could decide whether or not the group presentation was trivial, simply by trying all sequences shorter than a certain length.  Thus, the length of these paths must increase very rapidly for certain presentations of the trivial group.
Note that this example improves on the previous example, in that it might take very long for a path that simply winds once around a given $1$-cell to turn into the identity, and this path winding once around a given $1$-cell has non-negligible statistical weight, as opposed to the previous case where the path that wound many times around the neck of the dumbbell had very small statistical weight.

So, this gives us a sequence of examples where the diffusion of the paths becomes much slower than the diffusion of particles, despite having trivial $\pi_1$. 
There is still one unsatisfactory aspect of this example.  Namely, the sequence of paths to deform a given generator to the identity might involve increasing the length of the word by a large amount.  That is, the sequence of paths will be of the form $w_1,w_2,....,w_K$, where $w_1$ is one of the generators, $w_K$ is the identity, and $w_2,...,w_{N-1}$ are some other words in the group, and relations are used to move from $w_i$ to $w_{i+1}$, and where not only will $K$ become large, but possibly also the length of some of the $w_i$ will also become large.
As this length becomes large, the statistical weight of the corresponding path $w_i$ goes to zero.  This does of course even further slow the diffusion of paths, but we might be interested in an example where this weight does not go to zero as rapidly.

So an interesting question is whether there is a sequence of presentations of the trivial group for which we can find a sequence of paths from every generator to the identity, and for which the minimum length $K$ of such a sequence diverges rapidly, but for which every word $w_i$ in that sequence has at most polynomial length.
Note that in this case, $K$ is bounded by an exponential of a polynomial, because there are only that many possible words of polynomial length.

\subsection{Third Example}
Consider the group presentation with generators $g_1,...,g_n$ and relations
$g_1 g_2^{-2}, g_2 g_3^{-2}, ... , g_{n-1} g_n^{-2}, g_n$.  This is a presentation of the trivial group, as follows.  The last relation is the generator $g_n$.  We can multiply the second-to-last relation by $g_n^2$ to get the relation $g_{n-1}$ and then multiply the third-to-last relation by $g_{n-1}^2$ to get the relation $g_{n-2}$, and so on.

Further, given any word $w$, we can turn that word into the trivial word by a sequence of words using the relations without any word in that sequence having length more than $n$ longer than the original word.  Let us illustrate how to do this for the word $w=g_1$.  Then consider the following sequence of words: 
\begin{eqnarray}
\nonumber
&&
g_1, \; g_2^2, \; g_2 g_3^2, \; g_2 g_3 g_4^2,\;  ...,\; g_2 ... g_{n-1} g_n^2,\; g_2 ... g_{n-1},\\ \nonumber
&& g_2 ... g_{n-2} g_n^2, \; g_2 ... g_{n-2}, \; g_2 ... g_{n-3} g_{n-1}^2, \; g_2 ... g_{n-3} g_{n-1} g_n^2,
\end{eqnarray}
and so on.  At no point in this sequence does the length of the word get longer than $n$.  The algorithm to generate this sequence starting from $g_1$ is always to take the last generator in the sequence and if it is $g_n$ then remove that generator, but if it is some other generator $g_a$ to replace it with $g_{a+1}^2$.  We can use the same algorithm for any starting generator $g_a$.  Then, for any given starting word that is not a single generator, we can apply this sequence to turn the last generator in the original word into the identity, then the second-to-last, and so on.

The above sequence takes an exponential number of moves to turn $g_1$ into the identity.  We can prove that any sequence requires an exponential number of moves as follows.
Consider any word written as $g_{a_1}^{p_1} g_{a_2}^{p_2} ... g_{a_m}^{p_m}$, for some integer $m$.  Define the weight of the word to be
\be
\sum_{i=1}^m p_i (2^{n-a_i+1}-1).
\ee
This weight is invariant under conjugating the word by any generator, or by inserting a generator and its inverse anywhere, as one adds generators to the word with both positive and negative powers when doing this and these terms cancel in the weight.  Using a relation can change the weight by $\pm 1$.
The word $g_a$ has a weight $2^{n-a+1}-1$, so it takes at least $2^{n-a+1}-1$ moves to turn that word into the identity.
Note also that even if we add in the relations $g_a g_b g_a^{-1} g_b^{-1}$, specifying that all generators commute, we still have the same lower bound on the number of moves.

\begin{figure}
\includegraphics[width=2in]{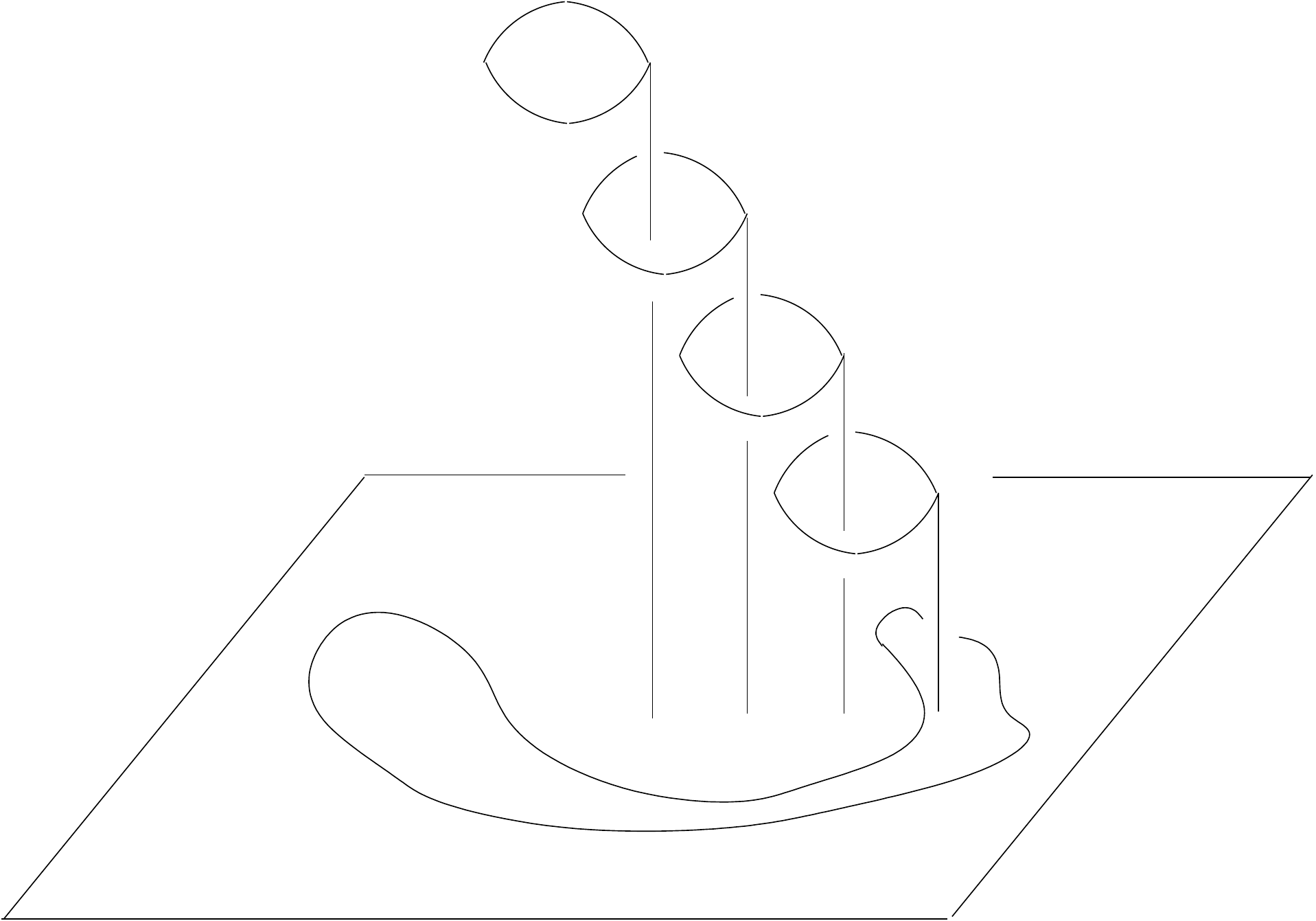}
\caption{The Desperado puzzle.  The puzzle rests on a flat plane, shown.  The objective is to separate the string from the wood pillars.   A ``topologist's solution" to the puzzle is to deform the pillars instead of the string, which makes it apparent that the string is in a topologically trivial configuration.  However it requires an exponential time to move the string between different configurations.}
\label{desp}
\end{figure}

Interestingly, this group presentation was inspired by considering a puzzle made of wood and string, called the Desperado puzzle\cite{desperado}.  See Fig.~\ref{desp}.  A similar effect occurs there, though that puzzle involves a space that is not simply connected.

Now, having constructed this group presentation, we construct an example where QMC will fail to find the ground state.  This is similar to the ``third example" previously using a bouquet of circles in subsection \ref{ThirdOne}, but here we have a trivial fundamental group.  Let ${\cal C}$ be the complex constructed above corresponding to this group presentation.  Let ${\cal C}'$ be a copy of ${\cal C}$.  We attach the two complexes to each other at a point by identifying the $0$-cell in ${\cal C}$ with that in ${\cal C}'$.  This corresponds to a group presentation with generators $g_1,...,g_n,g'_1,...,g'_n$ and relations
\begin{eqnarray}
\nonumber
&& g_1 g_2^{-2}, g_2 g_3^{-2}, ... , g_{n-1} g_n^{-2}, g_n,
\\ \nonumber
&& g'_1 (g'_2)^{-2}, g'_2 (g'_3)^{-2}, ... , g'_{n-1}( g'_n)^{-2}, g_n.
\end{eqnarray}  
We can also add in relations  $g_a g_b g_a^{-1} g_b^{-1}$ and  $g'_a g'_b (g'_a)^{-1}( g'_b)^{-1}$ if we choose, but we will not add in any relation $g_a g'_b g_a^{-1} (g'_b)^{-1}$.
Having defined this complex, we then subdivide the complexes and define a quantum Hamiltonian $H_{complex}$.
Each basis state that $H_{complex}$ acts on corresponds to some $0$-cell in the subdivided complex.  Let $|0\rangle$ be the $0$-cell in the subdivision that
corresponds to the $0$-cell in the original complex.
Then, add in an additional state $|r \rangle$ and consider the Hamiltonian
\begin{eqnarray}
H &=& mH_{bouquet} - h |0\rangle\langle 0| \\ \nonumber
&&- t |0\rangle \langle r| + h.c. \\ \nonumber
&&+E|r\rangle\langle r|.
\end{eqnarray}

We follow the same annealing protocol as in the ``third example" with the bouquet of circles, though the specific values we pick for $h,t,E$ will be different.
Consider what happens for the QMC protocol.  Initially, we have $c_i=r$ for all $i$, corresponding to the state $|r \rangle$.  However, once $t$ becomes non-zero, the trajectory starts to spend time on the bouquet.
Suppose at some pair of time slices $i,j$, we have $c_i=c_j=r$, but for all times $k \in \{i+1,i+2,...,j-1\}$ we have $c_k \neq r$.  Then, the sequence of configurations $c_{i+1},...,c_{j-1}$ is some sequence that starts and ends at $|0\rangle$ and forms a path.  This path corresponds to a word in the generators.  Let us write this word as $w=w_1 w'_1 w_2 w'_2 ...$ where $w_a$ is a word in the generators $g_1,...,g_n$ and $w'_a$ is a word in the generators
$g'_1,...,g'_n$.  Now, every path is topologically trivial, because the presentation describes a trivial group.  However, certain paths, namely those in which any of the $w_a$ or $w'_a$ have exponentially high weight cannot occur without taking an exponential length of time in the QMC.  Arbitrarily, let us say that the weight is high if it is more than $2^{n/2}$.
The number of words of length $l$ can be very crudely estimated as $(4n)^l$, since each generator in the word can be any of the generators $g_a$ or $g'_a$ or any of the inverses of these generators.  Taking into account the ability to reduce a word by cancelling a generator against its inverse (say, $g_a g_a^{-1}$) slightly reduces the base of the exponent, by an amount which is $o(n)$.  However, if we restrict to words in which no $w_a$ or $w'_a$ has weight more than $2^{n/2}$ appears, then the number of such words is $c^l$ for some $c<4n$ (in fact, $c/4n$ converges to some number less than $1$ in the limit of large $n$).  So there are exponentially fewer such words.

So, by restricting to these configurations where each $w_a$ or $w'_a$ does not have high weight exponentially reduces the statistical weight of the sum over paths corresponding to words $w$ of length $l$.  Effectively, this means that the QMC algorithm sees that the spectrum of states in the complex does not start at zero energy (the true ground state energy of the Laplacian on the complex) but in fact starts at some higher energy.  So, similar to before, we can find an annealing protocol
for which the energy $E$ is greater than zero (so that the quantum annealing procedure finds the ground state on the complex) but the energy $E$ is less than the bottom of the spectrum that the QMC sees (so that the QMC algorithm instead produces the state $|r\rangle$ at the end of the annealing protocol).

We can also take this example and do something similar to what we did in the ``fourth example" previously in subsection \ref{FourthOne}, and add an expander graph in addition to the state $|r\rangle$ to construct an example which has a trivial fundamental group for which QMC fails even with open boundary conditions.

\section{Analytic Results on Equilibration}
We now provide some analytic results on equilibration, including both positive and negative results.
Consider a continuous time Markov dynamics with transition rates from state $c$ to state $d$ given by:
\be
i \neq j \; \rightarrow \;
T_{dc}=J_{dc} \exp(-(E_d-E_c)/2),
\ee
where $J$ is a symmetric matrix.  
Then,
\be
T_{cc}=-\sum_{d \neq c} J_{dc}\exp(-(E_d-E_c)/2).
\ee
These rates satisfy detailed balance with stationary distribution $P_c\propto \exp(-E_c)$.

To analyze equilibration, we make a non-unitary change of basis to transform $T$ to a symmetric matrix $L$ by right-multiplying by $\exp(-E_c/2)$ and left-multiplying by the inverse matrix.  The resulting matrix $L$ has matrix elements
\begin{eqnarray}
c \neq d \; \rightarrow \; L_{dc} &=&  \exp(E_d/2) T_{dc}\exp(-E_c/2) \\ \nonumber
&=& J_{dc}.
\end{eqnarray}

This matrix is symmetric and real and has the same eigenvalues as $T$.
Write $H=-L$.  Then, $H$ has at least one zero eigenvalue, and a bound on the second eigenvalue gives an upper bound on equilibration time (the equilibration time will be bounded by the inverse of this eigenvalue times the logarithm of the size of this matrix, and for our purposes this logarithm grows only polynomially in $K,N$ and $\beta$).
Note then that in this section, we are using $H$ to refer to a Hamiltonian defined from $H=-L$.  Previously we used $H$ to refer to a quantum Hamiltonian.
To avoid ambiguity, in this section we will write $H^{quantum}$ to refer to the quantum Hamiltonian considered previously.

For local Monte Carlo moves for our problem with statistical weight (\ref{statweight}), 
we write the statistical weight as an exponential of an energy (for our particular choice of statistical weight, some energies might be infinite as some transitions are forbidden but in this section we will simply replace those with very large energies).
Then, we can write $H$ as a Hamiltonian for a one-dimensional spin chain, with $K$ spins, one spin per time slice.
We write $c_i$ for the value of the configuration on the $i$-th time slice.
We take $J=\sum_i J_i$, where $J_i$ is some symmetric matrix which is supported on time slices $i-1,i,i+1$ and does not change the value of $c_{i-1}$ or $c_{i+1}$.
 Let $J_i(c'_i,c_i|c_{i-1},c_{i+1})$ denote matrix elements of $J_i$ from $c_i$ to $c'_i$ for given values of $c_{i-1},c_{i+1}$. 
We write $E$ as $E=\sum_i E_{i,i+1}(c_i,i+1)$.
Then we can write $H=\sum_i H_i$, where $[H_i,H_j]=0$ for $|i-j|>1$.  Note that $H_i$ and $H_{i+2}$ both have support on the $i+1$-th time slice, but they still commute.

\subsection{Lower Bounds on Equilibration Time}
We write an orthonormal basis of states as $|c_1,...,c_K\rangle$.
Consider this orthonormal set of states:
\begin{eqnarray}
|c_1\rangle & \equiv & \sum_{c_2,c_3,...,c_K} Z(c_1)^{-1/2} \exp\Bigl(-\frac{E(c_1,...,c_K)}{2}\Bigr) \\ \nonumber
&& \times | c_1,...,c_K\rangle,
\end{eqnarray}
where $E(c_1,...,c_K)=\sum_{i=1}^K E_{i,i+1}(c_i,c_{i+1})$ and
\be
Z(c_1)=\sum_{c_2,c_3,...,c_K}\exp(-E(c_1,...,c_K)).
\ee
Note that $c_1$ is not summed over in either equation.

Then for $c'_1 \neq c_1$,
\begin{eqnarray}
&&\langle c'_1 | H | c_1 \rangle \\ \nonumber
&=&
-Z(c_1)^{-1/2} Z_2(c'_1)^{-1/2} 
\sum_{c_2,c_3,...,c_K}
J(c'_1,c_1|c_K,c_2) \\ \nonumber && \times
\exp\Bigl(-\frac{E(c_1,...,c_K)+E(c'_1,...,c_K)}{2}\Bigr) \\ \nonumber
&\equiv & -\tilde J(c'_1,c_1).
\end{eqnarray}
Also,
\be
\label{shows}
\langle c_1 | H | c_{1} \rangle
=
\sum_{c'_1 \neq c_1} \tilde J(c'_1,c_1) \sqrt{\frac{Z(c'_1)}{Z(c_1)}}.
\ee
Let $\tilde H$ be the operator $H$ projected into the space of states $|c_1\rangle,...,|c_K\rangle$, so the above equations define the matrix elements of $\tilde H$.
Eq.~(\ref{shows}) shows that $\tilde H$ can still be regarded as arising from equilibration of some Markov dynamics with an effective
 $\tilde J(c'_1,c_1)$ playing the role of $J$ and with $-\log(Z(c_1)$ playing the role of the energy $E$.

Note that the Markov process is defined by an energy $E$ and a matrix $J$.  If we have two different Markov processes, both using the same energy $E$ but one using a matrix $J$ and the second using a matrix $J'$ such that every matrix element of $J'$ is greater than or equal to the corresponding matrix element of $J$, the corresponding Hamiltonians obey the inequality that
$H' \geq H$.  Above, we have an $\tilde H$ which is the Hamiltonian corresponding to a Markov dynamics with energy $\log(Z(c_1)$ and matrix $\tilde J$.
Let us use $D$ to denote the matrix with the same off-diagonal matrix elements as $H^{quantum}$ but which is zero on the diagonal.
Define a matrix $\tilde J'=-C D$ where $C$ is the smallest constant such that every entry of $\tilde J'$ is greater than or equal to the corresponding entry of $\tilde J$.

Let $E_0$ be the ground state energy of $H^{quantum}$.  We now show that in the limit of large $\beta$, 
\be
\label{sameH}
\tilde H'=C(H^{quantum}-E_0)+O(\exp(-\beta \Delta) d^N),
\ee
where $N$ is the number of sites and $d$ is the dimension on a single site.
By construction, the off-diagonal matrix elements of $\tilde H'$ are the same as the off-diagonal matrix elements of $H^{quantum}$, up to multiplication by a factor of $C$.  Now consider the diagonal matrix elements.  Let $\psi_0(c)$ be the ground state wavefunction.
In the large $\beta$ limit, $Z(c_1)=|\psi_0(c_1)|^2+O(\exp(-\beta \Delta) d^N)$.  This error term $O(...)$ becomes negligible for $\beta=\poly(N)$.  In many cases, this error bound $O(\exp(-\beta \Delta) d^N)$ is in fact a large over-estimate of the true error.
Then, we find that
\begin{eqnarray}
&& \langle \psi(c)| \tilde H' | \psi(c) \rangle \\ \nonumber
&=& -C\sum_{c'\neq c} \langle \psi(c')| H^{quantum}| \psi(c) \rangle  \frac{\psi_0(c')}{\psi_0(c)} \\ \nonumber
&&+O(\exp(-\beta \Delta) d^N) \\ \nonumber
&=& C (\langle \psi(c) | H^{quantum} | \psi(c) \rangle-E_0)+O(\exp(-\beta \Delta) d^N),
\end{eqnarray}
where we used the fact that
$\psi_0(c)$ is an eigenstate of $H^{quantum}$, so that $\sum_{c'\neq c} \langle \psi(c')| H^{quantum}| \psi(c) \rangle \psi_0(c')=(E_0-\langle \psi(c) | H^{quantum} | \psi(c) \rangle)\psi_0(c)$.
This shows Eq.~(\ref{sameH}).

So, the lowest eigenvalue of $\tilde H'$ equals $C \Delta$.  Since $\tilde H'$ upper bounds $\tilde H$, this gives an upper bound to the lowest eigenvalue of $\tilde H$.   
Note that the matrix elements of $\tilde J$ are upper bounded by the corresponding matrix elements of $J$.  So, $C$ is upper bounded by the smallest constant $c$ such that the matrix elements of $-c D$ are greater than or equal to the corresponding matrix elements of $J$.  For a natural choice of local updates in which we update from state $c$ to $c'$ if there is a term in the Hamiltonian connecting those two states, this constant $c$ will be of order unity.
So, in this case, a small eigenvalue of $H^{quantum}$ implies a small eigenvalue of $H$ and hence a slow relaxation.

This result is perhaps not surprising, though, as for some values of $\beta$ we can also derive this result using exponential decay of correlations.  If $H^{quantum}$ has a small gap $\Delta$ compared to $\beta$, then the equilibrium state has long-range correlation in imaginary time.
Using our interpretation of $H$ as a Hamiltonian for a spin chain, this corresponds to a long-range correlation in the spin chain.  Using the exponential decay of correlations\cite{expdecay} for a gapped $H$, we can bound the gap of $H$.  However, the explicit relation between $\tilde H'$ and $H^{quantum}$ here may be interesting.  Further, if $\Delta$ is small but $\beta \Delta$ is large, then there may {\it not} be long-range correlations in imaginary time for {\it diagonal} operators.  Consider a simple Hamiltonian on a single spin coupled to a magnetic field $h$: $H^{quantum}=h S^z$, with gap $\Delta=h>0$.  Then, for $\beta>>\Delta$, the QMC probability distribution is dominated by the trajectory with the spin pointing down in all time slices so there are in fact no long-range correlations.

\subsection{Lower Bounds on Eigenvalue}
We need some preliminaries first.  We first derive some lower bounds for eigenvalues of one-dimensional quantum spin chains in Eq.~(\ref{final}), which can be understood as some sort of renormalization procedure.  We then apply to the specific spin chain arising from the QMC dynamics.

We begin with the following result:
let $P_1,P_2$ be projectors.  Then, for any real number $x$ with $0 \leq x \leq 1$ we have:
\be
\label{P1P2bnd}
x P_1+P_2 \geq \frac{x}{1+x} (1-P_1) P_2 (1-P_1).
\ee

To prove Eq.~(\ref{P1P2bnd}), note that by Jordan's lemma we can find a basis such that both $P_1,P_2$ become block diagonal, with the blocks having size either one or two.  We prove this equation for each block.  In a block of size one, then $P_1,P_2$ in that block equal $0$ or $1$, giving $4$ different possibilities for that block.  One can explicitly check all four possibilities.  Now consider a block of size two.  We can write $P_1$ in this block as
\be
\begin{pmatrix}
1 & 0 \\
0 & 0
\end{pmatrix},
\ee
and $P_2$ in this block as
\be
\begin{pmatrix}
\cos^2(\theta) & \cos(\theta) \sin(\theta) \\
\cos(\theta) \sin(\theta) & \sin^2(\theta)
\end{pmatrix}.
\ee
Let $y=\frac{x}{1+x}$.
So, we need to check that
\be
\begin{pmatrix}
x+\cos^2(\theta) & \cos(\theta) \sin(\theta) \\
\cos(\theta) \sin(\theta) & (1-y) \sin^2(\theta)
\end{pmatrix}
\geq 0.
\ee
This matrix is Hermitian.  For the given choice $x$, it has a positive trace.  So, it suffices to check that the determinant is positive.  The determinant
equals
\be
\Bigl[ (x+\cos^2(\theta))(1-y)-\cos^2(\theta) \Bigr] \sin^2(\theta).
\ee
Note that $\sin^2(\theta) \geq 0$.  The quantity in brackets is equal to $x(1-y)-y \cos^2(\theta)$.  Since $\cos^2(\theta) \leq 1$, this quantity in brackets is greater than or equal to $x(1-y)-y$, which for the given choice of $y$ is equal to $0$.

For Eq.~(\ref{P1P2bnd}), it follows that
\be
\label{P1P2bnd2}
P_1+P_2 \geq \alpha P_1 + \frac{1-\alpha}{2-\alpha} (1-P_1) P_2 (1-P_1),
\ee
for all $0 \leq \alpha \leq 1$
or that
\be
\label{P1P2bnd3}
\frac{1}{2}P_1+P_2 \geq \frac{\alpha}{2} P_1 + \frac{1-\alpha}{3-\alpha} (1-P_1) P_2 (1-P_1).
\ee

Consider a one-dimensional Hamiltonian of $K$ sites for some $K$,
\be
H=P_1+P_2+...+ P_K,
\ee
where each $P_i$ is a projector and where
\be
\label{comm}
|i-j| >1 \; \rightarrow \; [P_i,P_j]=0.
\ee
Assume for simplicity that $K$ is even.  We identity the $K$-th and the $0$-th sites and the distance in the above equation should be taken with this periodic identification.
Note that one way for Eq.~(\ref{comm}) to hold is if $P_i$ acts only on the $i$-th and $i+1$-th site.  However, later we will consider a more general way in which this equation can hold.

Then,
\be
H=\sum_{i=0,2,...} \Bigl(\frac{1}{2} P_{i} + P_{i+1} + \frac{1}{2} P_{i+2}\Bigr),
\ee
where the sum is over all even $i$ less than $K$.
Write $Q_{i}=1-P_{i}$.
Note that
$\frac{1}{2} P_{i} + \frac{1}{2} P_{i+2} \geq \frac{1}{2} (1-Q_{i}Q_{i+2})$ and also
$(1-Q_{i}Q_{i+2}) \geq \frac{1}{2} (P_{i} + P_{i+2})$.  (This is the place where we use Eq.~(\ref{comm})).
So, by Eq.~(\ref{P1P2bnd3}),
\begin{eqnarray}
&& \frac{1}{2} P_{i} + P_{i+1} + \frac{1}{2} P_{i+2} \\ \nonumber
&\geq & \frac{1}{2}(1-Q_{i}Q_{i+2})+P_{i+1} \\ \nonumber
& \geq & \frac{\alpha}{2} (1-Q_{i} Q_{i+2}) + \frac{1-\alpha}{3-\alpha} Q_{i,i+1}Q_{i+2} P_{i+1} Q_{i+2} Q_{i} \\ \nonumber
& \geq & \frac{\alpha}{4} (P_{i}+P_{i+2}) + \frac{1-\alpha}{3-\alpha} Q_{i}Q_{i+2} P_{i+1} Q_{i+2} Q_{i}.
\end{eqnarray}

So,
\be
\label{final}
H \geq \frac{\alpha}{4} H_{even}+\frac{1-\alpha}{3-\alpha} \tilde H_{odd},
\ee
where
\be
H_{even}=\sum_{i=0,2,...} P_{i}
\ee
and
\be
\tilde H_{odd} = \sum_{i=0,2,...} Q_{i}Q_{i+2} P_{i+1} Q_{i+2} Q_{i}.
\ee
Note that $H_{even}$ and $\tilde H_{odd}$ commute with all $Q_i$.

We now apply these results to equilibration.
Let $\lambda_0$ be minimum over $i$ of the smallest non-zero eigenvalue of $H_i$.  This quantity $\lambda_i$ characterizes how quickly a time-slice can equilibrate to its neighboring time slices.  So,
$H \geq \lambda_0 \sum_i P_i$,
where $P_i=1-Q_i$ and $Q_i$ projects onto the zero eigenspace of $H_i$.
The operator $P_i$ is supported on sites $i-1,...,i+1$ but it does not change the value of $c_{i-1}$ or $c_{i+1}$.

Let us assume that for any given $c_{i-1}$ and $c_{i+1}$ that $P_i$ has a unique ground state on sites $i-1,i,i+1$.  Then we can write a basis for the
eigenspace in which all even $Q_i$ are equal to $1$ by states of the form
\be
\label{states}
\sum_{c_2,c_4,,...} \phi_2(c_1,c_2,c_3) \phi_4(c_3,c_4,c_5) ...
|c_1,c_2,...,c_K \rangle.
\ee
Note that $c_1,c_3,...$ are {\it not} summed over and there is exactly one such eigenstate per choice of $c_1,c_3,...$.  Here we have defined
\begin{eqnarray}
&&\phi_i(c_{i-1},c_i,c_{i+1})
=\frac{\exp\Bigl(-\frac{E_{i-1,i}(c_{i-1},c_i)+E_{i,i+1}(c_i,c_{i+1})}{2}\Bigr)}{Z_i(c_{i-1},c_{i+1})^{1/2}},
\nonumber
\end{eqnarray}
where
\begin{eqnarray}
&& Z_i(c_{i-1},c_{i+1})\\\nonumber &=&\sum_{c_i} \exp(-(E_{i-1,i}(c_{i-1},c_i)+E_{i,i+1}(c_i,c_{i+1}))).
\end{eqnarray}
We write the state in Eq.~(\ref{states}) as $|c_1,c_3,...\rangle$ in a slight abuse of notation (if both odd and even $c_i$ appear in a ket then it is a state of the whole system, but if only odd $c_i$ appear in the ket then it is a state of form (\ref{states})).

Now compute $\tilde H_{odd}$ for this spin chain.
The operator
$Q_{i}Q_{i+2} P_{i+1} Q_{i+2} Q_{i}$ is supported on sites $i-1,i,...,i+3$.  We evaluate its matrix element between two states $|c_1,c_3,...\rangle$ and $|c'_1,c'_3,...\rangle$ that agree on all sites except site $i+1$.  Since this matrix element only depends upon $c_{i-1},c_{i+1},c'_{i+1},c_{i+3}$, we will not write any other $c_j$.  For notational simplicity, let us fix $i=2$.  Then
for $c'_3 \neq c_3$ we have
\begin{eqnarray}
&&
\langle c_1,c'_3,c_5 | P_{i+1} | c_{1}, c_{3},c_{5} \rangle \\ \nonumber
&=&
-\Bigl(Z_2(c_1,c_3)
Z_4(c_3,c_5)\Bigr)^{-1/2}  \Bigl(Z_2(c_1,c'_3) Z_4(c'_3,c_5)\Bigr)^{-1/2} \\ \nonumber
&& \times
\sum_{c_2,c_4}
J(c'_3,c_3|c_2,c_4) \\ \nonumber && \times \Bigl\{
\exp\Bigl(-E_{1,2}(c_{1},c_2)-\frac{E_{2,3}(c_2,c'_3)+E_{2,3}(c_2,c_3)}{2}\Bigr) \\ \nonumber 
&& \times
\exp\Bigl(-E_{4,5}(c_{4},c_5)-\frac{E_{3,4}(c'_3,c_4)+E_{3,4}(c_3,c_4)}{2}\Bigr) \Bigr\} \\ \nonumber
&\equiv & -\tilde J(c'_3,c_3|c_1,c_5).
\end{eqnarray}
Also
\begin{eqnarray}
&&\langle c_1,c_3,c_5 | P_{i+1} | c_{1}, c_{3},c_{5} \rangle \\ \nonumber
&=&
\sum_{c'_3 \neq c_3} \tilde J(c'_3,c_3|c_1,c_5) \sqrt{\frac{Z_2(c_1,c'_3) Z_4(c'_3,c_5)}{Z_2(c_1,c_3) Z_4(c_3,c_5)}}.
\end{eqnarray}

This procedure can be regarded as a kind of renormalization procedure.  We have a new $\tilde H_{odd}$, which acts on a spin chain with half as many sites.  This $\tilde H_{odd}$ can still be regarded as arising from equilibration of some Markov dynamics with a ``renormalized"
 $\tilde J(c'_3,c_3|c_1,c_5)$ playing the role played by $J(c'_3,c_3|c_2,c_4)$ and with $-\log(Z_2(c_1,c_3))$ playing the role of the term in the energy which depends upon a pair of neighboring sites.

The calculation in this subsection is very similar to the one in the previous subsection, in that in both cases we calculated a renormalized $\tilde J$ from the original $J$, although here we use it to give an upper bond and in the previous case we used it to give a lower bound.  Eq.~(\ref{final}) lower bounds $H$ in terms of $H_{even}$ and $H_{odd}$.  We can iterate Eq.~(\ref{final}).  That is, we can apply this equation to the Hamiltonian $\tilde H_{odd}$ and so on defining a renormalization procedure that halves the number of sites at each step.
The constant factor $\frac{1-\alpha}{3-\alpha}$ leads to an exponential decrease in the Hamiltonian from one step to the next.  We can pick a small value of $\alpha$ to make this constant close to $1/3$.  Since the number of renormalization steps is only logarithmic in $K$, this constant factor produces only a polynomial decrease in $H$.

So, we can lower bound the lower non-zero eigenvalue of $H$ by the polynomial factor from $(\frac{1-\alpha}{3-\alpha})^{\log_2(K)}$, multiplied by the product of $\lambda_0$ over steps.  The bond becomes ineffective if this $\lambda_0$ becomes small.
To give an example of how $\lambda_0$ can become small, consider the following toy model.  We have configurations $c$ labelling angles on a circle, so $0 \leq c < 2\pi$.  Suppose that the statistical weight vanishes if $|c_i-c_{i+1}| > 0.26 \pi$ for any $i$ and is equal to $1$ otherwise.  Then, a possible trajectory is $c_1=0,c_2=\pi/4,c_3=\pi/2,c_4=3\pi/4,c_5=\pi,...$.  Another possible trajectory is $c_1=0,c_2=-\pi/4,c_3=-\pi/2,c_4=-3\pi/4,c_5=\pi,...$.  Both trajectories have the same $c_1$ and $c_5$, but if those $c_1,c_5$ are held fixed then local updates cannot move from one trajectory to the other.  In this case, we find that $\tilde J'$ has no matrix elements between $c_3=\pi/2$ and $c_3=-\pi/2$ because if we choose a value of $c_2$ such as $\pi/4$ which is consistent with $c_3=\pi/2$, then it is inconsistent with $c_3=-\pi/2$.

Note also that it may not be necessary to run the renormalization until $\log(K)$ steps, if $K$ is sufficiently large compared to $\beta$.  After a large number of steps, the Hamiltonian may approximately ``decouple" $H$ into a sum of single site Hamiltonians, as the statistical weight will be a product of $|\psi_0(c_i)|^2$ over sites that remain.

\section{Discussion}
The work in this paper attacks the question of the computational complexity of the adiabatic algorithm with no sign problem.  Without the ``no sign problem" restriction, Ref.~\onlinecite{equivalence} shows that the adiabatic algorithm is equivalent to the circuit model.  With this restriction, a natural conjecture is that the adiabatic algorithm can only solve problems in the complexity class BPP.  While we have no definite results on the complexity, we have shown that the simplest way to place the adiabatic algorithm in BPP, by path integral QMC with local updates, does not work.
We have shown that it is possible to have a path in parameter space of quantum systems with a spectral gap that is only polynomially small and which have no sign problem, but for which QMC has exponentially slow equilibration for the natural choice of annealing protocol.
While the existence of obstructions to equilibrating QMC based on a nontrivial fundamental group are well-known, for example when studying bosons moving on a torus, and much effort has been devoted to nonlocal updates which might alleviate these problems, we have shown much stronger effects using a fundamental group which is a free group on two or more generators.  These stronger effects prevent QMC from accurately calculating the ground state energy, even at large $\beta$.

Perhaps more surprisingly, we have shown that slow equilibration of QMC can happen even if the fundamental group is trivial.  These examples are still based on results in topology, though, as they exploit the connection between a group presentation and a simplicial complex.

Finally, we have provided some analytic results connecting the spectral gap of the quantum Hamiltonian to the relaxation.  Interestingly, this implies that our third example in section ~\ref{notopo} gives a Markov dynamics whose corresponding Hamiltonian (that is, the Hamiltonian defined from the Markov dynamics, rather than $H^{quantum}$) has an exponentially small spectral gap but has only short-range correlations for far separated spins, because far separated spins correspond to very different imaginary times.

{\it Acknowledgments---} I thank M. Freedman for pointing out the ``Desperado puzzle" example and for many useful discussions.  I thank A. Harrow for raising my interest in the problem studied in this paper, and I thank A. Harrow and Z. Wang for useful discussions.
 I thank D. Wecker for useful comments on a draft of this paper.

\onecolumngrid

\vfuzz2pt
\newpage
\appendix
\changepage{}{-20mm}{}{+10mm}{}{}{}{}{}

\section*{Appendix. Recursive group presentations from low dimensional topology}
\begin{center}{\bf by M. H. Freedman} \end{center}
\vspace{12pt}

The purpose of this short appendix is to give some geometric/topological context to the phenomenon (see Section 4 of the main paper) of group presentations which, in geometric language, contain short contractible loops $\gamma$ which bound only exponentially large area disks $\Delta$.  Our examples all have the property that while area $\Delta\sim e^{\operatorname{length}(\gamma)}$, $\Delta$ is ``thin'' in that $\gamma$ can be swept over $\Delta$ with only a linear increase in length.  This thinness property appears to be quite generic: it is the basis for many ``string puzzles'' \cite{desperado} and may have some yet unexploited implications in topology.

\section{Solenoid}
Let us start with the dyadic solenoid and compare with the trivial group presentation of Section 4C:
\begin{equation}\label{eqn:A1}
\{g_1,\ldots,g_n \mid g_1 = g_2^2,\ldots,g_{n-1} = g_n^2,g_n\}
\end{equation}
The solenoid $X$ is the continua $X = \cap_{i=0}^\infty S_i$, where each $S_i$ is a solid torus and $S_{i+1}$ is embedded (with no normal twists) in $S_i$, as shown in Figure \ref{fig:A1}.

\begin{figure}[hbpt]
\[\begin{xy}<5mm,0mm>:
(0,0)*\xycircle<95pt>{}
,(-0.25,0)="A"
,"A"+"A"="B"
,(0.5,1)+"A"*{\hcap[4]}
,(0.5,1.25)+"A"*{\hcap[5]}
,(0.5,2.25)+"A"*{\hcap[9]}
,(0.5,2.5)+"A"*{\hcap[10]}
,(-0.75,0.25)+"A"*{\xbendr[2.5]}
,(-0.875,0.125)+"A"*{\xbendr[2.75]}
,(0.5,1)+"A"*{\hcap[-6.5]}
,(0.5,1.25)+"A"*{\hcap[-7.5]}
,(-0.75,2.25)+"A"*{\xbendl[-2.5]}
,(-1,2.5)+"A"*{\xbendl[-3]}
,(-1.25,1.625)+"A"*{\xbendu[-2]}
,(-1.5,1.75)+"A"*{\xbendu[-2]}
,(-0.75,0)*{\xcaph[3]}
,(-0.5625,0)*{\xcaph[-2.25]}
\ar@{-}@`{(-3.0625,1)+"B"} (-3,1.5)+"B";(-2.925,0.625)+"B"
\end{xy}\]
\caption{}
\label{fig:A1}
\end{figure}

We are actually concerned with the finite stages $X_n = S_n\subset S^1\times B^2 = S_0\subset\mathbb{R}^3\subset S^3$, where $S^3$ denotes the $3$-sphere $=\mathbb{R}^3\cup\infty$.  For example, $X_3$ wraps eight times through the meridian $\gamma$ of $S_1$ as shown in Figure \ref{fig:A2}.

\begin{figure}[hbpt]
\[\begin{xy}<5mm,0mm>:
(0,0)*\xycircle<95pt>{}
,(-0.3125,0)*{\xcaph[1.25]}
,(-0.1875,0)*{\xcaph[-0.75]}
,(-0.5,0)*{\vcap[-2]}
,(-0.75,0)*{\vcap[-3]}
,(-0.5,0)*{\vcap[3.25]}
,(-0.75,0)*{\vcap[4.25]}
,(-1.125,0)*{\vcap[-4.5]}
,(-1.375,0)*{\vcap[-5.5]}
,(-1.125,0)*{\vcap[5.75]}
,(-1.375,0)*{\vcap[6.75]}
,(-1.75,0)*{\vcap[-7]}
,(-2,0)*{\vcap[-8]}
,(-1.75,0)*{\vcap[8.25]}
,(-2,0)*{\vcap[9.25]}
,(-2.375,0)*{\vcap[-9.5]}
,(-2.625,0)*{\vcap[-10.5]}
,(-2.625,0.5)*{\xbendu[-1]}
,(-2.375,0.375)*{\xbendu[-0.75]}
,(0,0.5)*{\xbendd}
,(0.125,0.625)*{\xbendd[1.25]}
\ar@{-}@`{(-4.25,1.5625)} (-4.75,1);(-4,1.675)
\ar@{-} (-1.06,1.6125);(-1.5,1.67)
\ar@{-} (-0.65,1.95);(-1.1,2.075)
\ar@{-}@`{(-2.8125,1.725)} (-2.5,1.725);(-3.125,1.61)
\ar@{-}@`{(-2.65,2.25)} (-2.25,2.25);(-2.825,2.19)
\end{xy}\]
\caption{}
\label{fig:A2}
\end{figure}

While $\pi_1(S^3\backslash(\gamma\cup X_n))$ is an extremely complicated nonabelian group (It is an amalgamated free product of $n-1$ copies of $\pi_1(S_i - S_{i+1})$, which can be computed by the Wertinger algorithm as:
\begin{equation}
\{x_1,x_2,y_1,y_2,y_3 \mid y_1^{-1}x_2^{-1}y_1x_1, y_3^{-1}x_1^{-1}y_3x_2, y_3x^{-1}y_1x_1, y_3^{-1}y_2^{-1}y_3y_1\}),
\end{equation}
the fundamental group of $(S^3\backslash X_n)\cong \mathbb{Z}$, the integers, since $X_n\subset\mathbb{R}^3$ is an unknotted solid torus.  This $\mathbb{Z}$ is precisely the same as the presentation (\ref{eqn:A1}) with the final relation $g_n$ omitted.  Now attach a $2$-handle ($B^2\times I$, $\partial B^2\times I$) to the meridian of $S_n$ (corresponding to the relation $g_n$).  This partially fills the toroidal ``hole'' so the ``hole'' is now just a $3$-ball $B$, and $Y:=(S^3\backslash X_n)\cup 2$-handle is the complement $\cprotect\overline{S^3\backslash B}$, also a $3$-ball.  In particular, $\pi_1(Y) = \{e\}$ with presentation (\ref{eqn:A1}).

We are \emph{not} required to use the metric from $S^3$.  Let us instead take  all units $\overline{(S_i - S_{i+1})}$ isometric with $\partial S_i\cong\partial S_{i+1}\cong S^1_{\text{unit}}\times S^1_{\text{unit}}$ and the final $S_n\cong S^1_{\text{unit}}\times B^2_{\text{unit}}$, products of unit circles and disks in the Euclidean plane.  Thus, any surface $(\Sigma, \partial)\subset(S_n,\partial S_n)$ representing the generator $\delta$ of $H_2(S_n,\partial S_n; Z)$ has area $\geq\pi$.  (Since the composition $\Sigma\hookrightarrow S_n\cong S^1_{\text{unit}}\times D^2_{\text{unit}}\overset{\text{proj}}\rightarrow D^2_{\text{unit}}$ is locally area non-decreasing.)  Similarly, any surface representing $k\delta$ must have area $\geq k\pi$.

Evidently, the linking number $L(\gamma, S_n) = 2^n$.  So for homological reasons, any disk $\Delta$ (or even any oriented surface) bounding $\gamma$ must contain a (possibly disconnected) subsurface representing $2^n\delta$, and hence $\operatorname{area}(\Delta)\geq 2^n\pi$.

On the other hand, the obvious planar disk $\Delta$ bounding $\gamma$ and cutting through $S_n$ in $2^n$ meridional disks $\delta_i$, $1 < i \leq 2^n$, can be deformed to $\Delta^\prime$ by sliding each $\delta_i$ along $S_n$ until it enters the $2$-handle spanning a meridian to $S_n$, to lie in $Y$.  Metrically $\Delta^\prime$ has $2^n$ ``thumbs'' of area $\geq\pi$ each and height $\leq\pi$.

\vspace{.2in}
\begin{figure}[hbpt]
\[\begin{xy}<5mm,0mm>:
(0,1.875)="E"
,(0,-0.125)+"E"*{\xbendl[-0.5]}
,(-0.25,-0.25)+"E"*{\xbendu-}
,(-0.25,-0.75)+"E"*{\xcapv@(0)}
,(-0.75,-1.25)+"E"*{\xbendu}
,(-0.75,-1.625)+"E"*{\xbendl[0.5]}
,(-0.75,-1.875)+"E"*{\xbendr[-0.5]}
,(-0.75,-2)+"E"*{\xbendd}
,(-0.25,-2.5)+"E"*{\xcapv@(0)}
,(-0.25,-3)+"E"*{\xbendd-}
,(0,-3.25)+"E"*{\xbendr}
,(-2.4275,-3.625)+"E"+"E"*{\Delta^\prime}
\end{xy}\hspace{2cm}
\begin{xy}<5mm,0mm>:
,(2.875,1.125)="E8"
,(1,-0.5)+"E8"*{\zbendv}
,(2.3125,-0.625)+"E8"+"E8"*[white]\frm<8pt>{*}
,(1.375,0.125)+"E8"*{\xbendd[0.25]}
,(0.0625,0.125)+"E8"*{\xcaph[2.75]@(0.4)}
,(0.5,-0.875)+"E8"*{\xcapv[0.25]@(0)}
,(1.5,-1)+"E8";(0.75,-1)+"E8"*\dir{}*\ellipse<7pt,4pt>:a(180){}*\ellipse<7pt,4pt>=:a(180){.}
,(1.625,-2)+"E8"+"E8"*[white]\frm<7pt>{*}
,(3,0.625)="E7"
,(0,0)+"E7"*{\sbendv}
,(0.0625,-0.125)+"E7"+"E7"*[white]\frm<8pt>{*}
,(1,-0.5)+"E7"*{\zbendv}
,(1.375,0.125)+"E7"*{\xbendd[0.25]}
,(0.0625,0.125)+"E7"*{\xcaph[2.75]@(0.4)}
,(0.625,0.3125)+"E7"+"E7"*[white]\frm<8pt>{*}
,(0.5,-0.5)+"E7"*{\xcapv@(0)}
,(1,-0.5)+"E7"*{\xcapv@(0)}
,(1.5,-1)+"E7";(0.75,-1)+"E7"*\dir{}*\ellipse<7pt,4pt>:a(180){}*\ellipse<7pt,4pt>=:a(180){.}
,(2.125,0.75)="E6"
,(0,0)+"E6"*{\sbendv}
,(0.0625,-0.125)+"E6"+"E6"*[white]\frm<10pt>{*}
,(1,-0.5)+"E6"*{\zbendv}
,(1.375,0.125)+"E6"*{\xbendd[0.25]}
,(0.0625,0.125)+"E6"*{\xcaph[2.75]@(0.4)}
,(0.75,0.5)+"E6"+"E6"*[white]\frm<11pt>{*}
,(0.5,-0.5)+"E6"*{\xcapv@(0)}
,(1,-0.5)+"E6"*{\xcapv@(0)}
,(1.5,-1)+"E6";(0.75,-1)+"E6"*\dir{}*\ellipse<7pt,4pt>:a(180){}*\ellipse<7pt,4pt>=:a(180){.}
,(1.375,1)="E5"
,(0,0)+"E5"*{\sbendv}
,(0.25,-0.1875)+"E5"+"E5"*[white]\frm<5pt>{*}
,(1,-0.5)+"E5"*{\zbendv}
,(1.375,0.125)+"E5"*{\xbendd[0.25]}
,(0.0625,0.125)+"E5"*{\xcaph[2.75]@(0.4)}
,(0.5,-0.5)+"E5"*{\xcapv@(0)}
,(1,-0.5)+"E5"*{\xcapv@(0)}
,(1.5,-1)+"E5";(0.75,-1)+"E5"*\dir{}*\ellipse<7pt,4pt>:a(180){}*\ellipse<7pt,4pt>=:a(180){.}
,(0.0625,0.875)="E4"
,(1,-0.5)+"E4"*{\zbendv}
,(2.125,-0.6875)+"E4"+"E4"*[white]\frm<5pt>{*}
,(0,0.125)+"E4"*{\xbendu[-0.25]}
,(1.375,0.125)+"E4"*{\xbendd[0.25]}
,(0.0625,0.125)+"E4"*{\xcaph[2.75]@(0.4)}
,(1,-0.6875)+"E4"*{\xcapv[0.625]@(0)}
,(1.5,-1)+"E4";(0.75,-1)+"E4"*\dir{}*\ellipse<7pt,4pt>:a(180){}*\ellipse<7pt,4pt>=:a(180){.}
,(1.1875,-2)+"E4"+"E4"*[white]\frm<6pt>{*}
,(-0.25,0.5)="E1"
,(0,0)+"E1"*{\sbendv}
,(1,-0.5)+"E1"*{\zbendv}
,(0,0.125)+"E1"*{\xbendu[-0.25]}
,(1.375,0.125)+"E1"*{\xbendd[0.25]}
,(0.0625,0.125)+"E1"*{\xcaph[2.75]@(0.4)}
,(0.5,-0.5)+"E1"*{\xcapv@(0)}
,(1,-0.5)+"E1"*{\xcapv@(0)}
,(1.5,-1)+"E1";(0.75,-1)+"E1"*\dir{}*\ellipse<7pt,4pt>:a(180){}*\ellipse<7pt,4pt>=:a(180){.}
,(-2.5,0.75)="E2"
,(0,0)+"E2"*{\sbendv}
,(1,-0.5)+"E2"*{\zbendv}
,(0,0.125)+"E2"*{\xbendu[-0.25]}
,(1.375,0.125)+"E2"*{\xbendd[0.25]}
,(0.0625,0.125)+"E2"*{\xcaph[2.75]@(0.4)}
,(0.5,-0.5)+"E2"*{\xcapv@(0)}
,(1,-0.5)+"E2"*{\xcapv@(0)}
,(1.5,-1)+"E2";(0.75,-1)+"E2"*\dir{}*\ellipse<7pt,4pt>:a(180){}*\ellipse<7pt,4pt>=:a(180){.}
,(-1.5,1.0625)="E3"
,(1,-0.5)+"E3"*{\zbendv}
,(1.375,0.125)+"E3"*{\xbendd[0.25]}
,(0.0625,0.125)+"E3"*{\xcaph[2.75]@(0.4)}
,(0.5,-0.5)+"E3"*{\xcapv@(0)}
,(1,-0.5)+"E3"*{\xcapv@(0)}
,(1.5,-1)+"E3";(0.75,-1)+"E3"*\dir{}*\ellipse<7pt,4pt>:a(180){}*\ellipse<7pt,4pt>=:a(180){.}
,(0,0);(1,0)*\dir{}*\ellipse<120pt,40pt>:a(270),^,:a(60){-}
,(2,0);(1,0)*\dir{}*\ellipse<120pt,40pt>:a(120),^,:a(270){-}
\end{xy}\hspace{2cm}
\begin{xy}<5mm,0mm>:
(-2.4,-0.8)*\ellipse<8pt,16pt>{~}
,(-1.65,-0.35)*\ellipse<8pt,16pt>_,=:a(180){-}*\ellipse<8pt,16pt>:a(90),^,:a(270){.}
,(-1,-0.125)*\ellipse<8pt,16pt>_,=:a(180){-}*\ellipse<8pt,16pt>:a(90),^,:a(270){.}
,(0,0)*\ellipse<8pt,16pt>_,=:a(180){~}
,(-3.3,-0.7)*{\pi}
,(-2.25,-2.75)*{\mbox{\tiny $2$-handle}}
,(0.5,-2.75)*{\mbox{\tiny $\partial S_1$}}
\ar@{-}@/^/ (-5.25,-0.47);(0,1.2)
\ar@{-}@/^/ (-4.3,-2.7);(0,-1.1)
\ar (-2.25,-2.5);(-2.25,-1.75)
\ar (0,-2.25);(-0.5,-1.5)
\end{xy}\]
\caption{}
\label{fig:A3}
\end{figure}

The thinness property originally deduced from the presentation (\ref{eqn:A1}) can be understood geometrically: although $\Delta^\prime$ has exponentially many ``thumbs'' of size $O(1)$, we may avoid stretching $\gamma$ (more than linearly) by passing it over the ``thumbs'' one at a time.

\section{Half gropes and Devil's staircase}

Here we describe a $2$-complex, sometimes called a \emph{half grope}, which is at the heart of the ``Desperado'' or ``Devil's ladder'' string puzzles.  We consider only genus one examples: Figure \ref{fig:B1} shows a half grope $H$ of height $n=4$ with a possible \emph{cap disk} indicated with dotted lines.

\vspace{.1in}
\begin{figure}[hbpt]
\[\begin{xy}<5mm,0mm>:
(0,0)="A1"
,(0,0)+"A1"*{\vcap[4]}
,(0,0)+"A1"*{\xbendd[-1]}
,(1.5,0)+"A1"*{\xbendu}
,(0,-0.5)+"A1"*{\xbendd}
,(1.5,-0.5)+"A1"*{\xbendu-}
,(0.5,-1)+"A1"*{\xcapv@(0)}
,(1.5,-1)+"A1"*{\xcapv@(0)}
,(0.5,-1.5)+"A1"*{\xcapv@(0)}
,(1.5,-1.5)+"A1"*{\xcapv@(0)}
,(0.5,-2)+"A1"*{\xcapv@(0)}
,(1.5,-2)+"A1"*{\xcapv@(0)}
,(2,-2.5)+"A1";(1,-2.5)+"A1"*\dir{}*\ellipse<14pt,7pt>{}
,(0,0.5)+"A1";(1,0.5)+"A1"*\dir{}*\ellipse<7pt,14pt>_,=:a(180){}*\ellipse<7pt,14pt>{.}
,(-1.5,-5)+"A1"+"A1"*{\gamma}
,(1.125,-1)+"A1"+"A1"*{1}
,(2.875,1)+"A1"+"A1"*{2}
,(4,0)+"A1"+"A1"*[white]\frm<3pt>{*}
,(3,-2)+"A1"+"A1"*[white]\frm<3pt>{*}
,(3.5,0.5)="B1"
,(0,0)+"B1"*{\hcap[4]}
,(-0.5,0)+"B1"*{\xbendl-}
,(-0.5,-1.5)+"B1"*{\xbendr}
,(-1,0)+"B1"*{\xbendl}
,(-1,-1.5)+"B1"*{\xbendr-}
,(-1.5,-0.5)+"B1"*{\xcaph@(0)}
,(-1.5,-1.5)+"B1"*{\xcaph@(0)}
,(-2,-0.5)+"B1"*{\xcaph@(0)}
,(-2,-1.5)+"B1"*{\xcaph@(0)}
,(-2.5,-0.5)+"B1"*{\xcaph@(0)}
,(-2.5,-1.5)+"B1"*{\xcaph@(0)}
,(-1.5,-1)+"B1";(-2.5,-1)+"B1"*\dir{}*\ellipse<7pt,14pt>{}
,(-0.5,-1)+"B1";(0.5,-1)+"B1"*\dir{}*\ellipse<14pt,7pt>:a(0),^,:a(180){}*\ellipse<14pt,7pt>{.}
,(-1,-3)+"B1"+"B1"*{3}
,(1,-1)+"B1"+"B1"*{4}
,(-2,-1)+"B1"+"B1"*[white]\frm<3pt>{*}
,(0,0)+"B1"+"B1"*[white]\frm<3pt>{*}
,(2,2)="A2"
,(0,0)+"A2"*{\vcap[4]}
,(0,0)+"A2"*{\xbendd[-1]}
,(1.5,0)+"A2"*{\xbendu}
,(0,-0.5)+"A2"*{\xbendd}
,(1.5,-0.5)+"A2"*{\xbendu-}
,(0.5,-1)+"A2"*{\xcapv@(0)}
,(1.5,-1)+"A2"*{\xcapv@(0)}
,(0.5,-1.5)+"A2"*{\xcapv@(0)}
,(1.5,-1.5)+"A2"*{\xcapv@(0)}
,(0.5,-2)+"A2"*{\xcapv@(0)}
,(1.5,-2)+"A2"*{\xcapv@(0)}
,(2,-2.5)+"A2";(1,-2.5)+"A2"*\dir{}*\ellipse<14pt,7pt>{}
,(0,0.5)+"A2";(1,0.5)+"A2"*\dir{}*\ellipse<7pt,14pt>_,=:a(180){}*\ellipse<7pt,14pt>{.}
,(1.125,-1)+"A2"+"A2"*{5}
,(2.875,1)+"A2"+"A2"*{6}
,(4,0)+"A2"+"A2"*[white]\frm<3pt>{*}
,(3,-2)+"A2"+"A2"*[white]\frm<3pt>{*}
,(5.5,2.5)="B2"
,(0,0)+"B2"*{\hcap[4]}
,(-0.5,0)+"B2"*{\xbendl-}
,(-0.5,-1.5)+"B2"*{\xbendr}
,(-1,0)+"B2"*{\xbendl}
,(-1,-1.5)+"B2"*{\xbendr-}
,(-1.5,-0.5)+"B2"*{\xcaph@(0)}
,(-1.5,-1.5)+"B2"*{\xcaph@(0)}
,(-2,-0.5)+"B2"*{\xcaph@(0)}
,(-2,-1.5)+"B2"*{\xcaph@(0)}
,(-2.5,-0.5)+"B2"*{\xcaph@(0)}
,(-2.5,-1.5)+"B2"*{\xcaph@(0)}
,(-1.5,-1)+"B2";(-2.5,-1)+"B2"*\dir{}*\ellipse<7pt,14pt>{}
,(-0.5,-1)+"B2";(0.5,-1)+"B2"*\dir{}*\ellipse<14pt,7pt>:a(0),^,:a(180){}*\ellipse<14pt,7pt>{.}
,(-1,-3)+"B2"+"B2"*{7}
,(1,-1)+"B2"+"B2"*{8}
,(-2,-1)+"B2"+"B2"*[white]\frm<3pt>{*}
,(0,0)+"B2"+"B2"*[white]\frm<3pt>{*}
,(4,4)="A3"
,(1,-4.75)+"A3"+"A3";(1,0.75)+"A3"+"A3",**\dir{.}
,(3,-4.75)+"A3"+"A3";(3,0.75)+"A3"+"A3",**\dir{.}
,(2,-2.5)+"A3";(1,-2.5)+"A3"*\dir{}*\ellipse<14pt,7pt>{}
,(2,0.375)+"A3";(1,0.375)+"A3"*\dir{}*\ellipse<14pt>:a(0),^,:a(180){.}
\ar (-1,-5)+"A1"+"A1";(0,-5)+"A1"+"A1"
\end{xy}
\begin{xy}<5mm,0mm>:
(0,4.75)="E"
,(0,0)+"E"*{\xbendr[-2]}
,(0.5,-0.5)+"E"*{\xbendd[2]}
,(1.5,-1.5)+"E"*{\xcapv[2]@(0)}
,(1.5,-2.5)+"E"*{\xbendd[-2]}
,(2,-3.25)+"E"*{\xbendr}
,(2,-3.75)+"E"*{\xbendl-}
,(1.5,-4)+"E"*{\xbendu[-2]}
,(1.5,-5)+"E"*{\xcapv[2]@(0)}
,(0.5,-6)+"E"*{\xbendu[2]}
,(0,-6.5)+"E"*{\xbendl[2]}
,(6,-7.5)+"E"+"E"*{H_4^+}
\end{xy}\]
\caption{}
\label{fig:B1}
\end{figure}

The building block is a punctured torus.  If we take the puncture to be the size and shape of a longitude circle we may glue $n$ copies together as shown to produce $H_n$.  $H_n^+$ is the half grope union a final disk bounding the topmost longitude.  Using $[a,b]$ to represent $bab^{-1}a^{-1}$ and simply integers to denote generators, we may present $\pi_1(H_4)$ and $\pi_1(H_4^+)$ as follows.
\begin{equation}\label{eqn:B1}
\pi_1(G_4) = \{1,2,3,4,5,6,7,8 \mid 1=[3,4], 3=[5,6], 5=[7,8]\} \text{, and}
\end{equation}
\begin{equation}\label{eqn:B2}
\pi(G_4^+) = \{1,2,3,4,5,6,7,8 \mid 1=[3,4], 3=[5,6], 5=[7,8], 7=e\}
\end{equation}
and similarly for all $H_n$ and $H_n^+$.

The loop $\gamma = [1,2]$ and the relations tell us immediately that $\gamma$ lies in the $n$-stage of the lower central series of $\pi_1(H_n)\cong\Free(1, 2, 3, 4, \ldots, 2n)$, where we consider ordinary commutators to be in stage $1$ of the l.~c.~s.  On the other hand, $\gamma = e\in\pi_1(H_n^+)\cong\Free(2, 4, 6,\ldots, 2n)$.  To see that $\pi_1(H_n)$ and $\pi_1(H_n^+)$ are free, observe that they collapse to one-dimensional graphs, e.g., a punctured torus $H_1$ collapses to a wedge of two circles.

\begin{proposition}\label{prop:B1}
Any map $f$ of a disk bounding $\gamma$ into $H_n^+$ must pass over the cap at least $2^n$ times, i.e., if $p$ is the origin of the cap, $f^{-1}(p)$ must consist of at least $2^n$ points, which we may assume to be transverse.
\end{proposition}

We need:

\begin{lemma}\label{lem:B2}
Setting $\gamma = \partial H_1^+$ and $f:(D^2,\partial)\rightarrow (H^+, \gamma)$, $f$ $1$-$1$ on $\partial$, then $|f^{-1}(p)|\geq 2$.
\end{lemma}

\begin{proof}
It is readily computed (by a Mayer-Vietoris sequence) that $H_2(H_1^+, \gamma; Z)\cong Z$ and that $\partial : H_2(H_1^+,\gamma;Z)\rightarrow H_1(\gamma;Z)$ is an isomorphism.  Consequently, any two null homologies of $\gamma$ are themselves homologous (up to sign): $[f(D^2)] = \pm [H_1]\in H_2(H_1^+, \gamma; Z)$.  Since $H_1$ is disjoint from $p$, the homological intersection number $\sharp(f(D^2), p) = |H_1\cap p| = 0$.  But $\gamma$ is homotopically essential in $H_1$, so $|{f^\prime}^{-1}(p)| > 0$, for any $f^\prime$ homotopic to $f$, for if $f^\prime$ misses $p$ it may be deformed into $H_1$.  Since the signed sum of inverse images for $f^\prime$ generic is $0$, $|f^{-1}(p)|\geq 2$.
\end{proof}

For our induction we actually require a slightly stronger:

\begin{lemma}\label{lem:B3}
Let $f:(P,\partial P) \rightarrow (H_1^+,\gamma)$ be a map of a compact planar domain inducing degree $=\pm 1$ on $\partial$, $f_\ast [\partial P] = \pm 1\in H_1(\gamma; Z)\cong Z$, then $|f^{-1}(p)|\geq 2$.
\end{lemma}

\begin{proof}
The only new point is to show that the image $f(p)$ cannot lie in the punctured torus, $H_1$.  If $f$ did factor through $H_1$ then $f:(P,\partial B)\rightarrow (H_1, \partial)$ is a degree one map.  Let $\alpha$, $\beta$ be the dual meridian and longitude loops on $H_1$, respectively, and let $a$ and $b$ be their transverse inverse images $a = f^{-1}(\alpha)$, $b = f^{-1}(\beta)$.  Applying the $\operatorname{degree}(f) = 1$ property to the single transverse intersection $x = \alpha\cap\beta$ we see that $|f^{-1}(x)| =$ intersection number $(a, b) = 1$, contradicting the planarity of $P$.
\end{proof}

Apply \ref{lem:B3} to $H_k$ with all higher stages pinched to a disk to form $H_k^+$, starting with $k = 1,2,\ldots$.  Corresponding to the (at least) two points of opposite sign comprising $f^{-1}(p)$, $p\in$ cap $D^2$ of $H_1$ will contain (at least) two disjoint planar domains $P_+$ and $P_-$ mapping with opposite orientation over the cap of $H_1^+$.  These two planar domains can now be regarded as mapping into $H_2^+\backslash H_1)$, degree one on the boundary of the second stage.  \ref{lem:B3} now identifies further planar subdomains $P_{++}$, $P_{+-} \subset P_+$ and $P_{+-}$, $P_{--}\subset P_-$ mapping with opposite signs over the cap of $H_2^+$.  By induction we obtain $2^n$ disjoint planar domains $P_{n\text{-string}} \subset D$, each mapping over the final cap of $H_n^+$.  The orientation of each mapping is the weight of the string.  This proves \ref{prop:B1}. \qed

\ref{prop:B1} implies any disk in $H_n^+$ bounding $\gamma$ has exponential area $::2^n$.  Because the $\pi_1(H_n^+)$ is nonabelian, an algebraic---weight base argument---for this area estimate is not easy.

However, the proof of thinness for a suitably proven null homotopy of $\gamma$ is easily given in the algebraic context.  We simply illustrate the initial steps for shrinking $\gamma$ in $H_4^+$ with only a linear increase of its length:

\noindent $\gamma\rightarrow 1 2 \bar{1}\bar{2}\rightarrow 3 4 \bar{3}\bar{4}2\bar{1}\bar{2}\rightarrow 5 6 \bar{5}\bar{6} 4 \bar{3}\bar{4} 2 \bar{1}\bar{2}\rightarrow 7 8 \bar{7}\bar{8} 6 \bar{5}\bar{6} 4 \bar{3}\bar{4} 2 \bar{1}\bar{2}\rightarrow 8 \bar{7}\bar{8} 6 \bar{5}\bar{6} 4 \bar{3}\bar{4} 2 \bar{1}\bar{2}\rightarrow 6 \bar{5}\bar{6} 4 \bar{3}\bar{4} 2 \bar{1}\bar{2} \rightarrow 6\bar{7}\bar{8} 7 8 \bar{6} 4 \bar{3}\bar{4} 2 \bar{1}\bar{2} \rightarrow 6 \bar{8} 7 8 \bar{6} 4 \bar{3}\bar{4} 2 \bar{1}\bar{2}\rightarrow 4 \bar{3}\bar{4} 2 \bar{1}\bar{2}\rightarrow 4 \bar{5}\bar{6} 5 6 \bar{4} 2 \bar{1} \bar{2}\rightarrow 4 \bar{7}\bar{8} 7 8 \bar{6} 5 6 \bar{4} 2 \bar{1} \bar{2} \rightarrow 4 \bar{8} 7 8 \bar{6} 5 6 \bar{4} 2 \bar{1}\bar{2} \rightarrow 4 \bar{6} 5 6 \bar{4} 2 \bar{1}\bar{2} \rightarrow 4 \bar{6} 7 8 \bar{7}\bar{8} 6 \bar{4} 2 \bar{1}\bar{2}\rightarrow 4 \bar{6} 8 \bar{7}\bar{8} 6 \bar{4} 2 \bar{1}\bar{2} \rightarrow 2 \bar{1} 2 \rightarrow 2 \bar{3}\bar{4} 3 4 2 \rightarrow 2 5 6 \bar{5}\bar{6}\bar{4} 3 4 2 \rightarrow$ etc.,
where we always use a relation to increase the leftmost odd letter until it reaches $7$ ($ = 2n - 1$) and may be canceled.

The procedure above is identical to the YouTube video \cite{desperado} showing how to solve Puzzle Master ``Desperado.''

The essential features of the problem are still present in a simplified picture where the ambient fundamental group is only Z: Consider a slab in $\mathbb{R}^3$ with an unknotted but geometrically interesting arc $\alpha$ joining top to bottom.  We draw $\alpha$ below (Figure \ref{fig:B2}) so that the loop $\gamma$, also illustrated, bounds an embedded $H_4^+$ in the complement of $\alpha$.  (Find it!)  By a mild extension of the arguments used to prove Proposition \ref{prop:B1}, it may also be proved that the area of the smallest disk $\Delta\subset\text{slab}\backslash\alpha$ with $\partial\Delta = \gamma$ also grows exponentially with the number $n$ of self-feeding stages ($n = 4$ in Figure \ref{fig:B2}).  Of course the presence of $H_n^+$ confirms that $\gamma$ can be homotoped to a point so that its length increases only linearly with $n$.

\begin{figure}[hbpt]
\[\begin{xy}<5mm,0mm>:
(0,0)*{\hcap}
,(-0.5,-0.5)*{\xcaph@(0)}
,(-0.75,-0.5)*{\xbendl[-0.5]}
,(-0.875,-0.625)*{\xbendu[-0.5]}
,(-1,0)*{\xbendl-}
,(-1.25,-0.25)*{\xbendu-}
,(-0.875,-0.875)*{\xcapv@(0)}
,(-1.25,-0.75)*{\xcapv[1.25]@(0)}
,(-0.875,-1.5)*{\vcap[-0.875]}
,(-0.4375,-0.5625)*{\xcapv[1.875]@(0)}
,(-0.8125,-0.9375)="E1"
,(0,0)+"E1"*{\hcap}
,(-0.5,-0.5)+"E1"*{\xcaph@(0)}
,(-0.75,-0.5)+"E1"*{\xbendl[-0.5]}
,(-0.875,-0.625)+"E1"*{\xbendu[-0.5]}
,(-1,0)+"E1"*{\xbendl-}
,(-1.25,-0.25)+"E1"*{\xbendu-}
,(-0.875,-0.875)+"E1"*{\xcapv@(0)}
,(-1.25,-0.75)+"E1"*{\xcapv[1.25]@(0)}
,(-0.875,-1.5)+"E1"*{\vcap[-0.875]}
,(-0.4375,-0.5625)+"E1"*{\xcapv[1.875]@(0)}
,"E1"+"E1"="E2"
,(0,0)+"E2"*{\hcap}
,(-0.5,-0.5)+"E2"*{\xcaph@(0)}
,(-0.75,-0.5)+"E2"*{\xbendl[-0.5]}
,(-0.875,-0.625)+"E2"*{\xbendu[-0.5]}
,(-1,0)+"E2"*{\xbendl-}
,(-1.25,-0.25)+"E2"*{\xbendu-}
,(-0.875,-0.875)+"E2"*{\xcapv@(0)}
,(-1.25,-0.75)+"E2"*{\xcapv[1.25]@(0)}
,(-0.875,-1.5)+"E2"*{\vcap[-0.875]}
,(-0.4375,-0.5625)+"E2"*{\xcapv[1.875]@(0)}
,"E1"+"E2"="E3"
,(0,0)+"E3"*{\hcap}
,(-0.5,-0.5)+"E3"*{\xcaph@(0)}
,(-1,0)+"E3"*{\hcap-}
,(-1,0)+"E3"*{\xcaph@(0)}
,(-1,-0.5)+"E3"*{\xcaph@(0)}
,(-0.4375,-0.5625)+"E3"*{\xcapv[1.875]@(0)}
,(0.8125,0.9375)="D1"
,(0,0)+"D1"*{\hcap}
,(-0.5,-0.5)+"D1"*{\xcaph@(0)}
,(-0.75,-0.5)+"D1"*{\xbendl[-0.5]}
,(-0.875,-0.625)+"D1"*{\xbendu[-0.5]}
,(-1,0)+"D1"*{\xbendl-}
,(-1.25,-0.25)+"D1"*{\xbendu-}
,(-0.875,-0.875)+"D1"*{\xcapv@(0)}
,(-1.25,-0.75)+"D1"*{\xcapv[1.25]@(0)}
,(-0.875,-1.5)+"D1"*{\vcap[-0.875]}
,(-0.4375,-0.5625)+"D1"*{\xcapv[1.875]@(0)}
,"D1"+"D1"="D2"
,(0,0)+"D2"*{\hcap}
,(-0.5,0)+"D2"*{\xcaph@(0)}
,(-0.5,-0.5)+"D2"*{\xcaph@(0)}
,(-0.75,-0.5)+"D2"*{\xbendl[-0.5]}
,(-0.875,-0.625)+"D2"*{\xbendu[-0.5]}
,(-1,0)+"D2"*{\xbendl-}
,(-1.25,-0.25)+"D2"*{\xbendu-}
,(-0.875,-0.875)+"D2"*{\xcapv@(0)}
,(-1.25,-0.75)+"D2"*{\xcapv[1.25]@(0)}
,(-0.875,-1.5)+"D2"*{\vcap[-3]}
,(0.625,0.25)+"D2"*{\xcapv[3.5]@(0)}
,(0.625,1)+"D2"*{\xcapv@(0)}
,(-6,0.375)+"D2"*{\xcaph[17]@(0)}
,(-6,0.375)+"D2"-(0,7)*{\xcaph[17]@(0)}
,(-8.75,-9.5);(-7.5,-7.5),**\dir{-}
,(8.25,-9.5);(9.5,-7.5),**\dir{-}
,(-8.75,4.5);(-7.5,6.5),**\dir{-}
,(8.25,4.5);(9.5,6.5),**\dir{-}
,(-7,-5.125)*{\gamma}
,(-0.25,-3)*{\alpha}
,(6,-3)*{\pi_1(\text{slab}\backslash\alpha)\cong Z}
\end{xy}\]
\caption{}
\label{fig:B2}
\end{figure}

In this geometry, clearly loops with zero winding defuse only slowly (while points defuse quickly).  Since $\text{slab}\verb|\|\alpha$ is homeomorphic a solid torus, it is natural to wonder if the circular coordinate---even though not geometrically a Cartesian product---might be interpreted as imaginary time in an exotic finite temperature path integral.

\section{Gropes}

As a final example we complete \emph{half gropes} $H_n$, $H_n^+$ to \emph{gropes} $G_n$ and $G_n^+$.  If caps are present, this leads to interesting presentations of the trivial group, associated now to the commutator series (rather than the l.~c.~s.~of Section B).  The presentation of $\pi_1(G_n^+)\cong \{e\}$ has exponentially (in $n$) many generators and relators; however, each generator is the consequence of only linearly many relators, via the many half groups $H_n^+\subset G_n^+$.  We start with a geometric picture of $G_3^+$
\begin{figure}[hbpt]
\vspace{.3in}
\[\begin{xy}<5mm,0mm>:
(0,0)="A"
,(0,0)+"A"*{\vcap[12]}
,(0,0)+"A"*{\xbendd[-3]}
,(4.5,0)+"A"*{\xbendu[3]}
,(0,-1.5)+"A"*{\xbendd[3]}
,(4.5,-1.5)+"A"*{\xbendu[-3]}
,(1.5,-3)+"A"*{\xcapv[3]@(0)}
,(4.5,-3)+"A"*{\xcapv[3]@(0)}
,(1.5,-4.5)+"A"*{\xcapv[3]@(0)}
,(4.5,-4.5)+"A"*{\xcapv[3]@(0)}
,(2,-6)+"A";(3,-6)+"A"*\dir{}*\ellipse<42pt,21pt>{}
,(1.8,1.15)+"A";(2.8,1.45)+"A"*\dir{}*\ellipse<21pt,42pt>_,=:a(180){}*\ellipse<21pt,42pt>{.}
,(1.5,0.5)+"A"*\ellipse<8pt>:a(0),^,:a(180){.}*\ellipse<8pt>:a(180),^,:a(0){-}
,(2.25,0)+"A";(1.25,0)+"A"*\dir{}*\ellipse<4pt,8pt>:a(90),^,:a(270){.}*\ellipse<4pt,8pt>:a(270),^,:a(90){-}
,(0.96875,0)+"A"*\ellipse<1.5pt>{-}
,(0.91,0)+"A"*{\vloop[-0.225]}
,(0.83,0)+"A"*\ellipse<2.25pt,1.5pt>{-}
,(0.745,0)+"A"*{\vloop[0.33]}
,(1.5,0.9325)+"A"*\ellipse<1.5pt,2.25pt>{-}
,(1.5,1.0625)+"A"*\ellipse<2.25pt,1.5pt>{-}
,(1.21875,0.5)+"A"*{\vloop[1.125]}
,(1.25,0.28125)+"A"*{\hloop[-1.125]}
,(1.8,0.45)+"A"+"A"*\cir<21pt>{dl^r}
,(12,0)+"A"+"A"*[white]\frm<3pt>{*}
,(9,-6)+"A"+"A"*[white]\frm<3pt>{*}
,(2,-12)+"A"+"A"*{\mbox{\huge $\gamma$}}
,(9,1.5)="B"
,(0,0)+"B"*{\hcap[12]}
,(-1.5,0)+"B"*{\xbendl[-3]}
,(-1.5,-4.5)+"B"*{\xbendr[3]}
,(-3,0)+"B"*{\xbendl[3]}
,(-3,-4.5)+"B"*{\xbendr[-3]}
,(-4.5,-1.5)+"B"*{\xcaph[3]@(0)}
,(-4.5,-4.5)+"B"*{\xcaph[3]@(0)}
,(-6,-1.5)+"B"*{\xcaph[3]@(0)}
,(-6,-4.5)+"B"*{\xcaph[3]@(0)}
,(-7,-3)+"B";(-6,-3)+"B"*\dir{}*\ellipse<21pt,42pt>{}
,(1,-3)+"B";(2,-3)+"B"*\dir{}*\ellipse<26pt,13pt>:a(0),^,:a(180){.}*\ellipse<26pt,13pt>:a(180),^,:a(0){-}
,(-0.9325,-3.2)+"B";(0.0625,-3)+"B"*\dir{}*\ellipse<30pt,15pt>:a(0),^,:a(180){.}
,(1.0625,-2.8)+"B";(0.0625,-3)+"B"*\dir{}*\ellipse<30pt,15pt>:a(180),_,:a(0){-}
,(1.3725,-2)+"B"*{\vcap[2.5]}
,(0.875,-1.2)+"B";(-0.125,-1.5)+"B"*\dir{}*\ellipse<18pt,12pt>:a(0),^,:a(180){-}
,(0,-2.125)+"B"*\ellipse<6pt>{-}
,(0,-1.59375)+"B"*\ellipse<6pt,10pt>{-}
,(2,-2.25)+"B"*\ellipse<4pt>{-}
,(2,-2.11)+"B"*{\hloop[0.575]}
,(3,-1.875)+"B";(2,-1.875)+"B"*\dir{}*\ellipse<3.5pt,6.5pt>{-}
,(2,-1.64)+"B"*{\hloop[-0.925]}
\ar@{-}@`{(3.5,4)+"A"+"A",(4.375,5.6875)+"A"+"A"} (0.75,1.5)+"A"+"A";(4.75,5.6875)+"A"+"A"
\ar@{-}@`{(5.5,0.25)+"A"+"A"} (1.8125,-1.03125)+"A"+"A";(6.1875,0.0625)+"A"+"A"
\ar@{-}@`{(2.75,-4.75)+"B"+"B"} (2.75,-4)+"B"+"B";(2.125,-6)+"B"+"B"
\ar@{-}@`{(5.25,-4.75)+"B"+"B"} (5.25,-4)+"B"+"B";(5.8125,-6)+"B"+"B"
\ar@{-}@`{(1.25,-4)+"B"+"B"} (1,-2.8125)+"B"+"B";(2.25,-5.75)+"B"+"B"
\ar@{-}@`{(-1.5,-4.25)+"B"+"B"} (-1.475,-3.1875)+"B"+"B";(-1.9325,-6.25)+"B"+"B"
\ar@{-}@`{(2.625,2)+"A"+"A",(2.75,1.875)+"A"+"A"} (2.875,2.25)+"A"+"A";(3,2)+"A"+"A"
\ar@{-}@`{(3.375,1.75)+"A"+"A",(3.125,1.5)+"A"+"A"} (3.125,2)+"A"+"A";(2.875,1.78125)+"A"+"A"
\ar@{-}@`{(-1.25,-3.5)+"B"+"B",(-1,-4)+"B"+"B"} (-0.28125,-2.53125)+"B"+"B";(0,-3.8125)+"B"+"B"
\ar@{-}@`{(2,-4.25)+"B"+"B",(2,-5)+"B"+"B"} (0,-3.8125)+"B"+"B";(0,-4.6875)+"B"+"B"
\end{xy}
\begin{xy}<5mm,0mm>:
(-4,4)="E"
,(0,0)+"E"*{\xbendr[-3]}
,(0.75,-0.75)+"E"*{\xbendd[3]}
,(2.25,-2.25)+"E"*{\xcapv[3]@(0)}
,(2.25,-3.75)+"E"*{\xbendd[-3]}
,(3,-4.875)+"E"*{\xbendr[1.5]}
,(3,-5.625)+"E"*{\xbendl[-1.5]}
,(2.25,-6)+"E"*{\xbendu[-3]}
,(2.25,-7.5)+"E"*{\xcapv[3]@(0)}
,(0.75,-9)+"E"*{\xbendu[3]}
,(0,-9.75)+"E"*{\xbendl[3]}
,(8.5,-11.25)+"E"+"E"*{G_3^+}
\end{xy}\]
\vspace{.2in}
\caption{}
\label{fig:C1}
\end{figure}
where some of the caps and even the final surface stages have become too small (in the illustration) to draw carefully.  In order to maintain the correspondence between area and number of group relations we should actually think of each punctured torus $T$ piece of each stage as the same size and shape.  In words start with a punctured torus $T_0$, glue $T_{00}$ and $T_{01}$ to its meridian and longitude.  Then continue to strings of length $n$ gluing $T_{n\text{-string}}$.  This produces $G_n$.  To produce $G_n^+$, continue by adding $2^n$ disk to the meridian and longitudes of the top stage.

The presentations are:
\begin{equation}\label{eqn:C1}
\begin{aligned}
\hspace{1cm}
\pi_1(G_n) = & \{0,1,00,01,10,11,\ldots,\overbrace{11\cdots 1}^n \mid \text{each string of length $< n$} \\
& \text{ is the commutator of its two extensions}\}
\end{aligned}
\end{equation}
$\pi_1(G_n^+) = \{\text{above presentations} + \text{the relations} : \text{all } n \text{-strings are trivial}\} = \{e\}$, the trivial group.

Analogously to the results of Section B we have:

\begin{proposition}\label{prop:C1}
Any map of a disk into $G_n^+$ bounding $\gamma$ must pass over at least $2^n$ caps (counted with multiplicity) and therefore have area exponential in $n$. \qed
\end{proposition}

\begin{proposition}\label{prop:C2}
$\gamma$ bounds thin disks $\Delta$ mapping into $G_n^+$ in the sense that $\gamma$ may be homotoped to a point along $\Delta$ without ever increasing its length more than linearly in $n$.
\end{proposition}

Like Section A the space with poorly diffusing loops is now simply connected.  Unlike Section A the ``parent group,'' $\pi_1(G_n)$, the group before adding the trivializing relations, is nonabelian.

A final remark.  Looking for impediments to loop diffusion has brought us into the heart of \emph{wild topology}.  Consider an infinite grope $G_\infty$ with geometrically shrinking stages.  Take a tapered neighborhood $\mathcal{N}(G_\infty)$ which becomes thinner out toward the higher stages and complete with the dyadic Cantor set of limit point to $G_\infty$.  This closed neighborhood $\overline{\mathcal{N}}(G_\infty)$ is nothing other than the famous Alexander horned ball, the exotic closed complementary region of Alexander's ``horned'' embedding of the $2$-sphere into the $3$-sphere.


\begin{thebibliography}{99}
\bibitem{CC}  The Cheeger constant is a lower bound on the surface-to-volume ratio of a region of less than half the volume of the space.  A lower bound on the first non-zero eigenvalue of the Laplacian can be obtained from the Cheeger constant.
\end{thebibliography}
\end{document}